\documentclass[preprint,12pt]{elsarticle}

\usepackage{graphicx,amssymb,amsmath,amsthm,upgreek}
\graphicspath{{EPS/}}
\graphicspath{{EPS/}}
\biboptions{comma,square,compress}

\newcommand{\dd}{\mathrm{d}}
\newcommand{\D}{\mathrm{D}}
\newcommand{\e}{\mathrm{e}}
\newcommand{\x}{\mathbf{x}}
\newcommand{\ii}{\mathrm{i}}
\newcommand{\kk}{\mathbf{k}}
\newcommand{\upi}{\uppi}
\newcommand{\sgn}{\mathrm{sgn}}
\newcommand{\oo}{\mathrm{o}}
\newcommand{\OO}{\mathrm{O}}
\newcommand{\CBF}{\mathfrak{Q}}
\newcommand{\LICM}{\mathfrak{L}}
\newcommand{\const}{\mathrm{const}}

\newtheorem{theorem}{Theorem}[section]
\newtheorem{lemma}[theorem]{Lemma}
\newtheorem{corollary}[theorem]{Corollary}
\newtheorem{proposition}[theorem]{Proposition}
\newtheorem{definition}[theorem]{Definition}
\newtheorem{remark}{Remark}

\journal{Wave Motion}

\begin{document}

\begin{frontmatter}

\title{Wave propagation in linear viscoelastic media with completely monotonic relaxation moduli}

\author{Andrzej Hanyga}

\address{ul. Bitwy Warszawskiej 1920 r. 14/52\\
02-366 Warszawa, Poland}

\begin{abstract}
It is shown that viscoelastic wave dispersion and attenuation in a viscoelastic medium with a
completely monotonic relaxation modulus is completely characterized by the phase speed and 
the dispersion-attenuation spectral measure. The dispersion and attenuation functions are 
expressed in terms of a single dispersion-attenuation spectral measure. An alternative expression of 
the mutual dependence of the dispersion and attenuation functions, known as the 
Kramers-Kronig dispersion relation, is also derived  
from the theory. The minimum phase aspect of the filters involved in the Green's function is another 
consequence of the theory. Explicit integral expressions for 
the attenuation and dispersion functions are obtained for a few analytical relaxation models.
\end{abstract}


\end{frontmatter}

\section*{Notation}
{\footnotesize
\begin{tabular}{l l l}
$ \overline{z}$ & complex conjugate of $z$&\\
$ ]a,b]$ &  & $\{ x \mid a < x \leq b \}$\\
$\{a\}$ & one point set & \\
$  t_+^{\;\alpha}$ & & $ t_+^{\;\alpha} = t^\alpha$ for $t>0$ and 0 otherwise\\
$ \chi_{[a,b]}$ & characteristic function of the segment $[a,b]$ & \\
$ f(x) \sim_a g(x)$  & asymptotic equivalence & $\lim_{x\rightarrow a}\,[f(x)/g(x)] = 1$\\
$ f(x) = \OO_a[g(x)]$  &    & $0 < \lim_{x \rightarrow a}\,[f(x)/g(x)] < \infty$\\
$ f(x) = \oo_a[g(x)]$  &   & $\lim_{x \rightarrow a}\,[f(x)/g(x)] = 0$\\
$ \D^n\, f(x)$ & derivative & $\D^n\, f(x) = \dd^n\, f(x)/\dd x^n$\\
$ f^\prime(x)$ &   & $f^\prime(x) = \D f(x)$ \\
$u_{,t} := \partial u/\partial t$ & & \\
$ \mathbb{R}, \mathbb{R}_+, \mathbb{C}$ & real, positive real, complex numbers & $\mathbb{R}_+ := \{ x \in \mathbb{R} \mid x > 0\}$\\
$ \mathbb{C}_+$ & right half complex plane & $\mathbb{C}_+ = \{ z \in \mathbb{C} \mid \vert \arg z\vert < \upi/2\}$\\
$ \mathbb{C}^\pm$ & upper/lower half complex plane & $\mathbb{C}^\pm := \{ z \in \mathbb{C} \mid \pm \Im z > 0\}$\\
$ L f = \tilde{f}$ & Laplace transform of $f$ & $\tilde{f}(y) = \int_0^\infty f(x) \, \exp(- y x)\, \dd x$\\
$L \mu$ & Laplace transform of a Radon measure & $(L \mu)(x) := \int_{[0,\infty[} \e^{- x y} \, \mu(\dd y)$.
\end{tabular}
}
\section{Introduction.}

The assumption that the relaxation modulus of a viscoelastic medium is a completely monotonic function 
has many implications for dispersion, attenuation and the Kramers-Kronig (K-K) dispersion relations. These 
aspects are very important for acoustics, in particular for the ultrasound applications. Ultrasound 
attenuation and mechanical tests are used as complementary methods of investigating molecular 
relaxation in polymers and soft matter \cite{AligAl88}. Interconversion of ultrasonic data and mechanical test data is 
important for improving the reliability of the results. In materials with dielectric properties molecular 
relaxation can be investigated by studying dielectric loss. Dielectric loss in polar dielectrics exhibits 
similar features (alpha and beta peaks) to the mechanical loss modulus while the dielectric permittivity,
representing gradual polarization of the material in an electric field, has all the properties of mechanical creep
\cite{HanSerDielectrics07}.  

In seismological applications 
the minimum-phase aspect of viscoelastic wave propagation is relevant for deconvolution. 

The assumption that the relaxation modulus is a completely monotonic function has its roots in the interpolation of 
experimental data in terms of Prony sums with positive coefficients \cite{Bland:VE,ParkSchapery99}. By Bernstein's theorem 
complete monotonicity of the relaxation modulus is equivalent to the statement that the relaxation spectral measure is 
non-negative. A Prony sum with positive coefficients is a completely monotonic function with a spectral measure supported by a finite set of points. Several attempts have been made to justify complete monotonicity of the relaxation modulus by an 
assumption about dependence of stress on loops in the past history \cite{Day70a}, by an assumed relation between time 
dependence of strain and stress 
\cite{BerisEdwards93} or by other arguments \cite{AnderssenLoy02b}. The theory of viscoelasticity based on the assumption of 
complete monotonicity of the relaxation modulus is very attractive because it leads to a fairly complete characterization
of the creep compliance, complex modulus, attenuation and dispersion as well as the possible anisotropic
properties of the medium \cite{HanDuality}. A weak point of the theory is an instability of the property of complete monotonicity:
in an arbitrary neighborhood of a completely monotonic function in the space of bounded continuous functions
there are functions which are not completely monotonic. The assumption of complete monotonicity has therefore 
to be considered as an a priori restriction on the data and on the interpolating functions. Parameter estimation 
based on experimental data has to be performed in the class of completely monotonic functions. This is however implicit 
in routine modeling of experimental data for stress relaxation in terms of Prony sums for general viscoelastic media \cite{ParkSchapery99} or in terms of Cole-Cole \cite{ColeCole,BagleyTorvik3}, Havriliak-Negami 
\cite{HavriliakHavriliak,Boyd85,AligAl88} and Kohlrausch-Williams-Watts \cite{AlvarezAl93} relaxation moduli for more specific classes of viscoelastic materials. The same relaxation functions are used in modeling dielectric relaxation in the same
materials if they exhibit dielectric properties \cite{HanSerDielectrics07,CapelasAl2011}.

Completely monotonic relaxation moduli and the associated complex moduli 
are determined by a relaxation spectral measure, which represents the weight of participating Debye relaxation mechanisms 
(Sec.~\ref{sec:constitutive}). It is shown below that the complex wave number function, the attenuation function and the 
dispersion function of a medium with a completely monotonic relaxation modulus are parameterized by the dispersion-attenuation 
spectral measure (Sec.~\ref{sec:CWF}). 

The dispersion function and the attenuation function satisfy a dispersion relation in parametric form because both functions are 
expressed in terms of a single spectral measure. This dispersion relation implies the K-K dispersion relations
with two subtractions (Sec.~\ref{sec:KK}). The K-K dispersion relations indicate that the complex wave number 
function is the Laplace transform 
of a causal distribution. This distribution turns out to be a second-order distributional derivative
of a causal function defined in terms of the Laplace transform of the dispersion-attenuation spectral measure. 

The K-K dispersion relations are closely related to the minimum phase property of viscoelastic Green's functions 
(Sec.~\ref{sec:MPH}). This property is relevant for deconvolution of seismic signals.

The dispersion-attenuation spectral measure can often be explicitly calculated for a given analytic complex modulus by analytic continuation
to the entire complex plane cut along the negative real semi-axis and calculating the jump of the analytic continuation 
at the branch cut. This procedure is demonstrated for the Cole-Cole, Havriliak-Negami and Cole-Davidson  relaxation models (Sec.~\ref{sec:examples}). The dispersion-attenuation spectral 
measure provides an efficient tool for numerical determination of the dispersion and attenuation functions.

\section{Mathematical preliminaries.}

\begin{definition}
A function $f$ defined on the open positive real semi-axis $\mathbb{R}_+$ is said to be 
{\em completely monotonic (CM)} if it has derivatives of arbitrary high order and 
$$(-1)^n \D^n\, f(x) \geq 0 \qquad\text{for $x > 0$ and $n = 0,1,2, \ldots $}$$
\end{definition}
A CM function can have a singularity at 0. A bounded CM function $f$ has a limit at 0 
and its domain of definition can therefore be extended by continuity to the closed
positive real semi-axis. A CM function is locally integrable if and only if it is
integrable over $[0,1]$. 

\begin{theorem}\label{thm:Bernstein} (Bernstein's Theorem)
Every CM function $f$ has an integral representation 
\begin{equation}\label{eq:realBernstein}
f(x) = \int_{[0,\infty[} \e^{-r x} \, \mu(\dd r), \qquad t > 0
\end{equation} 
where $\mu$ is a positive Radon measure (that is, a locally finite measure), such that
$$\int_{[0,\infty[} \e^{-r \varepsilon} \mu(\dd r) < \infty $$
for some $\varepsilon > 0$.

The Radon measure $\mu$ is uniquely defined by $f$. 
\end{theorem}

\begin{definition}
$\mathfrak{M}$ is the set of positive Radon measures $\mu$ satisfying the inequality 
\begin{equation} \label{eq:doss}
\int_{]0,\infty[} \frac{\mu(\dd r)}{1 + r} < \infty
\end{equation}
\end{definition}

\begin{theorem} \label{thm:locint}
A real-valued CM function is locally integrable if and only if the Radon 
measure $\mu$ in eq.~\eqref{eq:realBernstein} belongs to $\mathfrak{M}$.
\end{theorem}
The proof is given in Appendix~\ref{app:locint}.
Locally integrable CM functions will be denoted by the abbreviation LICM and the set of LICM functions
will be denoted by $\LICM$. 

\begin{remark} \label{rem:1}
Inequality \eqref{eq:doss} is satisfied if and only if $\int_{]0,1]} \mu(\dd r) < \infty$ and
$\int_{]1,\infty[} \mu(\dd r)/r < \infty$ hold simultaneously. This statement follows immediately from
the inequalities $1/2 \leq 1/(1 + r) \leq 1$ and $1/(2 r) \leq 1/(1 + r) \leq 1/r$ holding on $[0,1]$ and
on $[1,\infty[\,$, respectively. By a similar argument \eqref{eq:doss} is equivalent to the inequality
$$\int_{]0,\infty[} \frac{\mu(\dd r)}{a + r} < \infty$$ for an arbitrary number $a > 0$. 
\end{remark}

\begin{definition}
A function $f$ defined on the closed positive real semi-axis is said to be a {\em Bernstein function}
if $f \geq 0$ and $\D f$ is completely monotonic.
\end{definition}
Since a Bernstein function $f$ is non-decreasing and non-negative, the limit $\lim_{x \rightarrow 0+} f(x)$ always 
exists and is finite. 

A superposition of Bernstein functions with positive coefficients is a Bernstein function.
The set $\mathfrak{B}$ of Bernstein functions is closed under pointwise convergence \cite{BernsteinFunctions}. Consequently an integral of 
a family of Bernstein functions with a non-negative weight is a Bernstein function.

A function $g$ is a Bernstein function if and only if it has the form
\begin{equation} \label{eq:BernsteinFunction}
g(x) = a + \int_0^x h(y) \, \dd y
\end{equation}
where $a \geq 0$ and $h$ is a LICM function. Indeed, $g^\prime = h$ is a CM function and $g \geq 0$, hence $g$ 
is a Bernstein function.
On the other hand, if $g$ is a Bernstein function then its derivative $g^\prime$ is a CM and it is 
locally integrable. Hence eq.~\eqref{eq:BernsteinFunction} holds with $h := g^\prime$, $a = g(0)$. In particular 
the primitive function $1 - \e^{-x}$ of the CM function $\e^{-x}$ is a Bernstein function. 

\begin{definition} \label{def:CBF}
A function $f: \mathbb{R}_+ \rightarrow \mathbb{R}$ is a {\em complete Bernstein function} (CBF)
if there is a Bernstein function $g$ such that $f(x) = x^2\, \tilde{g}(x)$.
\end{definition}
The set of complete Bernstein functions will be denoted by the symbol $\CBF$. 

Let $g$ be the Bernstein function \eqref{eq:BernsteinFunction}.\\
The Laplace transform of $g$ has the form $\tilde{g}(x) = a/x + \tilde{h}(x)/x$, 
where $a \geq 0$ and $h$ is a LICM function. 
By Bernstein's theorem 
$h(y) = \int_{[0,\infty[} \e^{-y r}\, \mu(\dd r)$, where $\mu \in \mathfrak{M}$. Hence
$$\tilde{h}(x) = L^2\, \mu = \int_{[0,\infty[} \frac{\mu(\dd r)}{x + r}$$
If $f$ is a CBF and $f(x) = x^2 \, \tilde{g}(x)$, then 
$$
f(x) = a\, x + x \int_{[0,\infty[} \frac{\mu(\dd r)}{x + r}, \qquad x > 0
$$
or, equivalently, 
\begin{equation} \label{eq:CBFintegral}
f(x) = a\, x + b + x \int_{]0,\infty[} \frac{\mu(\dd r)}{x + r}, \qquad x > 0
\end{equation}
where $\mu \in \mathfrak{M}$, $a \geq 0$ and $b := \mu(\{0\}) \geq 0$. 
Any function with an integral representation of the form \eqref{eq:CBFintegral} 
with $a, b \geq 0$ and $\mu \in \mathfrak{M}$ is obviously a CBF. In particular, the function
$x/(x + a)$ is a CBF if $a \geq 0$. 

We now note that $x/(x + a)$ is a Bernstein function because it is an integral of Bernstein 
functions $x \rightarrow 1 - \e^{-r x}$, $r \geq 0$, with a positive weight $\e^{-a r}$:
$$ \frac{x}{x + a} = a \int_0^\infty \e^{-a r}\, \left(1 - \e^{-r x}\right) \, \dd r$$
Eq.~\eqref{eq:CBFintegral} now implies that every CBF function is an integral of Bernstein 
functions with a positive weight, hence it is a Bernstein function.

The Bernstein function $f(x) := 1 - \e^{-x}$ is however not a CBF. Indeed, suppose the contrary.
Then $f(x)$ has the integral representation \eqref{eq:CBFintegral} with $b = f(0) = 0$
and $a = \lim_{x\rightarrow\infty} f(x)/x = 0$. Hence
$f(x)/x = \int_{]0,\infty[}  \mu(\dd r)/(x + r)$. But $f(x)/x = 
\int_0^\infty \e^{-x y}\,\chi_{[0,1]}(y) \, \dd y$ and therefore the characteristic function
$\chi_{[0,1]}(x) = \int_{]0,\infty[} \e^{-x y} \, \mu(\dd y)$ is a smooth CM function. 
This conclusion is false, hence $1 - \e^{-x}$ is not a CBF. Two examples of complete 
Bernstein functions relevant for attenuation and dispersion are $x^\alpha$ and $(1 + x)^\alpha - 1$,
$0 < \alpha < 1$.

Many other examples of CBFs can be found in \cite{BernsteinFunctions}.

Let $g(z)$ denote the analytic continuation of $g(x)$ to the complex plane cut along the negative real axis. 
Eq.~\eqref{eq:CBFintegral} implies that $\Im g(z) \geq 0$ in $\mathbb{C}^+$. By the
Pick-Nevanlinna theorem (\cite{BernsteinFunctions} Theorem~6.7) every non-negative continuous function $g(x)$ 
on $\overline{\mathbb{R}_+}$ which has an analytic continuation with
the above property is a CBF. This criterion allows identifying some complex analytic functions as 
analytic continuations of CBF, in particular $z^\alpha$ with $0 \leq \alpha \leq 1$ and $\ln(1 + z)$. 

We shall also need the following theorems \cite{BernsteinFunctions}. They follow easily from the 
Pick-Nevanlinna theorem. The first one is an immediate consequence of the Pick-Nevanlinna theorem.
\begin{theorem} \label{thm:CBF1}
If $f$ is a CBF and $0 < \alpha \leq 1$ , then $f(\cdot)^\alpha$ is a CBF. 
\end{theorem}
\begin{theorem} \label{thm:CBF2}
A function $f \not\equiv 0$ is a CBF if and only if the function $x/f(x)$ is a CBF. 
\end{theorem}
\begin{proof}
Let $g(x) := x/f(x)$.

If $f$ is a CBF then it has an analytic continuation $f(z)$ to $\mathbb{C}^+$. The analytic continuation of $f$  has 
an integral representation of the form \eqref{eq:CBFintegral}. It follows that  
$f(z)/z$ swaps $\mathbb{C}^+$ and $\mathbb{C}^-$. Consequently its inverse $z/f(z)$ maps 
$\mathbb{C}^+$ and $\mathbb{C}^-$ 
into itself and is non-negative on $\mathbb{R}_+$. Hence $g(x) \equiv x/f(x)$ is a CBF. 

Conversely, $f(x) = x/g(x)$. Consequently, if $g(x)$ is a CBF, then $f(x)$ is a CBF.
\end{proof}

For simplicity we shall henceforth apply the term CBF to the analytic continuation of a CBF as well. 

\section{Viscoelastic media with non-negative relaxation spectrum.}
\label{sec:constitutive}

We shall henceforth assume that the relaxation modulus $G$ is LICM \cite{SerHan2010}. Since the function $G$ is non-increasing
and non-negative, it has a limit $G_\infty := \lim_{t\rightarrow\infty} G(t) \geq 0$. The 
limit at 0 can be infinite even for physically realistic models such as the Rouse theory of dilute polymer solutions
\cite{ReHrNo:VE}.

According to eq.~\eqref{eq:realBernstein} 
$$G(t) = \int_{[0,\infty[} \e^{-r t} \, \mu(\dd r), \qquad t > 0$$
where $\mu \in \mathfrak{M}$. The relaxation modulus is thus a linear superposition of Debye relaxation
functions with non-negative weights. The {\em relaxation spectral measure} $\mu$ is positive. 

On account of eq.~\eqref{eq:CBFintegral} the function 
\begin{equation}
Q(p) := p \, \tilde{G}(p) = p \int_{[0,\infty[} \frac{\mu(\dd r)}{p + r}
\end{equation}
is a CBF. It follows from the general theory of the Laplace transform \cite{Doetsch,WidderLT} that
$Q(0) = G_\infty  \geq 0$ and, additionally, $\lim_{p\rightarrow\infty} Q(p) = G_0 := \lim_{t\rightarrow 0+} G(t)$ if
the relaxation modulus is bounded. 

\section{One- and three-dimensional viscoelastic Green's functions.}

Viscoelastic Green's functions are solutions of the problem
\begin{equation}
\rho \, u_{,tt} = G(t)\ast\nabla^2 \, u_{,t}, \qquad u(0,x) = 0, \quad u_{,t}(0,\x) = \delta(\x)
\end{equation}

The viscoelastic Green's function in a three-dimensional space can be expressed in terms of 
the one-dimensional Green's function:
\begin{equation} \label{eq:1Dto3D}
u^{(3)}(t,\x) =  -\left.\frac{1}{2 \upi r} \frac{\partial}{\partial r} u^{(1)}(t,r)\right\vert_{r=\vert \x \vert}
\end{equation}
Indeed, we have
$$
u^{(1)}(t,x) = \frac{1}{2 \upi} \frac{1}{2 \upi \ii} \int_{-\ii \infty + 
\varepsilon}^{\ii \infty + \varepsilon} 
\e^{p t} \, \frac{\rho}{Q(p)}\,\dd p \int_{-\infty}^\infty \e^{\ii k x} \frac{1}{k^2 + \kappa(p)^2} \dd k $$
hence
\begin{equation} \label{eq:Green1D}
u^{(1)}(t,x) = \frac{1}{4 \upi \ii} \int_{-\ii \infty + \varepsilon}^{\ii \infty + \varepsilon}
\e^{p \,t} \, \frac{\rho}{Q(p) \, \kappa(p)}\, \e^{-\kappa(p)\, \vert x \vert} \, \dd p
\end{equation}
where the {\em complex wave number function} $\kappa(p)$ is defined by the equation 
\begin{equation} \label{eq:kp}
\kappa(p) = \rho^{1/2}\,p/Q(p)^{1/2}
\end{equation} 
The square root is defined in such a way that $\Re \kappa(-\ii \omega) \geq 0$ for real $\omega$. 

In the three-dimensional case 
\begin{multline*}
u^{(3)}(t,\x) = \frac{1}{(2 \upi)^3} \frac{1}{2 \upi \ii} 
\int_{-\ii \infty + \varepsilon}^{\ii \infty + \varepsilon} \e^{p t} \, 
\frac{\rho}{Q(p)}\,\dd p 
\int_{\mathbb{R}^3} \dd_3 k \frac{\e^{\ii \kk\cdot\x}}{k^2 + \kappa(p)^2} = \\
\frac{1}{(2 \upi)^2} \frac{1}{2 \upi \ii} 
\int_{-\ii \infty + \varepsilon}^{\ii \infty + \varepsilon} \e^{p t} \, \frac{\rho}{Q(p)}\dd p 
\int_0^\infty k^2 \, \dd k \int_0^\upi \sin(\vartheta)\, \dd \vartheta 
\frac{\e^{\ii k \vert \x \vert \,\cos(\vartheta)}}{k^2 + \kappa(p)^2} 
\end{multline*}
Hence, by closing the contour over $k$ in the upper half of the complex $k$-plane 
\begin{multline*} 
u^{(3)}(t,\x) =
\frac{-1}{(2 \upi)^3\, \vert \x \vert} \int_{-\ii \infty + \varepsilon}^{\ii \infty + \varepsilon} 
\e^{p t} \, \frac{\rho}{Q(p)}\,\dd p \int_{-\infty}^\infty \, 
\frac{\e^{\ii k \vert \x \vert}}{k^2 + \kappa(p)^2}\, k \, \dd k = \\
\frac{1}{8 \upi^2 \,\ii\, \vert \x \vert} \int_{-\ii \infty + \varepsilon}^
{\ii \infty + \varepsilon} \dd p \,\frac{\rho}{Q(p)}
\e^{p\, t - \kappa(p) \, \vert \x \vert} = -\left. \frac{1}{2 \upi r} 
\frac{\partial}{\partial r} u^{(1)}(t,r)\right\vert_{r=\vert \x \vert}
\end{multline*}

\section{The complex wave number function.}
\label{sec:CWF}

Assume that under constant strain the stress does not relax to 0: $G_\infty := \lim_{t\rightarrow \infty} G(t) > 0$.

By Theorems~\ref{thm:CBF1} and \ref{thm:CBF2} the complex wave number function  $\kappa$ is a CBF. 
Since $Q(0) = G_\infty > 0$, $\kappa(0) = 0$. Hence $\kappa$ has an integral representation 
\eqref{eq:CBFintegral} without the constant term:
\begin{equation} \label{eq:inteRepres}
\kappa(p) = B \, p + \beta(p)
\end{equation}
where $B > 0$ and 
\begin{equation} \label{eq:inteRepres1}
\beta(p) := p \int_{]0,\infty[} \frac{\nu(\dd r)}{p + r}
\end{equation} 
and $\nu \in \mathfrak{M}$. The function $\beta$ will be called the {\em dispersion-attenuation function}.
The {\em dispersion-attenuation spectral measure} $\nu$ determines the attenuation, dispersion and wavefront
behavior of viscoelastic wave motion.

We have just shown that the complex wave number function of a viscoelastic medium with a LICM relaxation 
modulus is a CBF. The converse is not true.
Indeed, $Q(p) = \rho \, [p/\kappa(p)]^2$ is a square of the CBF function $\rho^{1/2}\, p/\kappa(p)$.
A square of a CBF need not however be a CBF, for example $p^{2/3}$ is a CBF but its square is not even a 
Bernstein function.

\begin{theorem} \label{thm:rebeta}
If $\Re p \geq 0$ then $\Re \beta(p) \geq 0$. 
\end{theorem}
\begin{proof}
If $\Re p \geq 0$ then 
\begin{equation} \label{eq:1}
\Re \beta(p) = \int_{]0,\infty[} \frac{\vert p \vert^2 + r\, \Re p}{\vert p + r\vert^2} \nu(\dd r) \geq 0
\end{equation}
\end{proof}

\begin{theorem} \label{thm:2}
$\Re \beta(p)$ is a non-decreasing function of $\vert p \vert$ in the closed right half plane $\overline{\mathbb{C}_+}$.
\end{theorem}
\begin{proof}
The integrand of eq.~\eqref{eq:1} has the form $f(x) := (x^2 + a x)/(x^2 + ax + c)$, where $x = \vert p \vert$,
$a : = r \cos(\arg(p)) \geq 0$, $c = r^2 \geq 0$. The derivative of $f$ is non-negative. 
\end{proof}

The derivative
$$\beta^\prime(p) = \int_{]0,\infty[} \frac{r\, \nu(\dd r)}{(p + r)^2} > 0$$
exists for real $p$ because the integrand is bounded by $\nu(\dd r)/(p + r)$.
We are however interested in growth of $\Re \beta(p)$ as $p$ tends to infinity in the closed right 
complex half-plane. 

\begin{theorem} \label{thm:2finitewavefrontspeed}
If $\nu([0,\infty[) = \infty$ then $\Re \beta(p) \rightarrow \infty$ as $\vert p \vert \rightarrow \infty$
in the right half of the complex $p$-plane.
\end{theorem}
\begin{proof}
The integrand of \eqref{eq:1} tends to 1 as $\vert p \vert \rightarrow \infty$. Setting $x = 1/\vert p\vert$,
$a := \cos(\arg(p))$ the integrand is transformed into the function
$(1 + a x)/(1 + 2 a x + x^2)$, which has a non-positive derivative. Hence the integrand of 
\eqref{eq:1} is monotonically increasing to 1. If $\nu([0,\infty[) = \infty$, then by the Lebesgue-Fatou lemma
\cite{Yosida} $\Re \beta(p)$ increases to infinity as $\vert p \vert \rightarrow \infty$, $p \in \mathbb{C}_+$.
\end{proof}

\begin{theorem} \label{thm:betaineq}
$\beta(p)/p$ tends to 0 as $\vert p \vert \rightarrow \infty$ in the right half-plane $-\upi/2 \leq \arg p \leq \upi/2$
uniformly with respect to $\arg p$.
\end{theorem}
\begin{proof}
For $-\upi/2 \leq \varphi \leq \upi/2$
$$\inf_{\varphi \in [-\upi/2,\upi/2]} \vert R\, \e^{\ii \varphi} + r\vert^2 = \inf_{\varphi \in [-\upi/2,\upi/2]} 
\left[ R^2 + r^2 + 2 R r \cos(\varphi)\right] = R^2 + r^2$$
hence
$$\sup_{\varphi \in [-\upi/2,\upi/2]} \left\vert \frac{\beta\left(R \e^{\ii \varphi}\right)}{R \e^{\ii \varphi}} \right\vert
\leq \int_{[0,\infty[} \frac{\nu(\dd r)}{\sqrt{R^2 + r^2}}$$
But for $R \geq 1/\sqrt{3}$ and $r \geq 1$ the inequality 
$1/\sqrt{R^2 + r^2} \leq \sqrt{3}/\sqrt{1 + 3 r^2}
\leq 1/(1 + r)$ holds 
in view of the inequality $(1 + r)^2 \leq 1 + 3 r^2$ for $r \geq 1$, while for $R \geq \sqrt{3}$ and $r < 1$, 
$1/\sqrt{R^2 + r^2} \leq 1/\sqrt{3 + r^2}
\leq 1/(1 + r)$. 
Thus
$1/\sqrt{R^2 + r^2} \leq 1/(1 + r)$ for $R \geq \sqrt{3}$ and $r \geq 0$. In view of eq.~\eqref{eq:doss}
$$\lim_{R \rightarrow \infty} \int_{[0,\infty[} \frac{\nu(\dd r)}{\sqrt{R^2 + r^2}} = 0$$
by the Lebesgue Dominated Convergence Theorem. Hence $\vert \beta(p)/p \vert$ tends to 0 in 
$\overline{\mathbb{C}_+}$ uniformly with respect to $\arg p \in [-\upi/2,\upi/2]$.
\end{proof}

\begin{corollary} \label{cor:B}
If $G_0 < \infty$ then $B = (\rho/G_0)^{1/2}$ in eq.~\eqref{eq:inteRepres}, otherwise $B = 0$.
\end{corollary}
\begin{proof}\mbox{ } \\
Recall that $\lim_{p\rightarrow\infty} Q(p) = G_0 \leq \infty$.

Theorem~\ref{thm:betaineq} implies that 
$$B = \lim_{ p\rightarrow \infty \atop \Re p \geq 0 } \frac{\kappa(p)}{p} =
 \lim_{p\rightarrow \infty \atop \Re p \geq 0} \frac{\rho^{1/2}}{Q(p)^{1/2}} =
\left(\frac{\rho}{G_0}\right)^{1/2}$$
\end{proof}

The phase function in the Green's function \eqref{eq:Green1D} has the form 
$-\ii \omega\, t - \kappa(-\ii \omega) r = -\ii \omega (t - B r) - \beta(-\ii \omega)\, r$ 
with $\beta(-\ii \omega) = \oo_\infty[\omega]$ and $\vert \exp(-\beta(p) \, r) \vert \leq 1$. 
It will be shown that Green's function vanishes for $t < B \, r$ if $B > 0$ and
thus $c_\infty := 1/B$ can be identified as the {\em wavefront speed}. The parameter $B$, $B \geq 0$, 
will henceforth be replaced by $c_\infty$ varying in the range $0 < c_\infty \leq \infty$.

Eq.~\eqref{eq:inteRepres} implies that the function $\kappa$ can be analytically continued to the complex plane cut along
the negative real axis, i.e. to the principal Riemann sheet $p \in \mathbb{C}$, $-\upi < \arg p < \upi$. The jump 
of the complex analytic function on the branch cut $]-\infty,0]$ can be easily calculated. Note that the function 
$\lambda(p) := \kappa(p)/p$ assumes the following boundary values on both sides of the branch cut:
$$\lambda\left(R \, \e^{\pm \ii \upi}\right) = B + \int_0^\infty \frac{\nu(\dd r)}{r - R \pm  \ii 0+}, \qquad R > 0$$
and 
$$(z \pm \ii 0+)^{-1} = \mathrm{vp} z^{-1} \mp \ii \upi \delta(z)$$
where $\mathrm{vp} \,z^{-1}$ denotes the principal value of $z^{-1}$ and $\delta$ denotes the Dirac delta
\cite{Gelfand}. Assume for a while that $\nu(\dd r) = h(r) \, \dd r$, where the density $h$ is a smooth
function. From the above identities follows the equation
\begin{equation}
f(R) := \left[\lambda\left(R \, \e^{\ii \upi}\right) - \lambda\left(R \, \e^{- \ii \upi}\right)\right]/(2 \ii \upi) = h(R)
\end{equation}
 
The boundary values of analytic functions are in general distributions \cite{BeltramiWohlers}. Hence the jump function 
$f$ is a distribution. Since
$f$ is a non-negative distribution, it is a measure. More generally the measure $\nu$ of a segment $]u,w]$ is given
by the formula 
\begin{equation}  \label{eq:jumpmeasure}
\nu(]u,w]) = \frac{1}{\upi} \int_u^w \Im \lambda\left(R \, \e^{\ii \upi}\right) \, \dd R
\end{equation}
(Theorem~\ref{thm:anCM}).

\section{A necessary and sufficient condition for complete monotonicity of the relaxation modulus.}

\begin{theorem}
The function $\kappa(p)$ of a viscoelastic material 
satisfies the condition $\kappa(p)^2/p \in \CBF$ if the relaxation modulus 
is completely monotonic.
\end{theorem}
\begin{proof}
Eq.~\eqref{eq:kp} implies that $Q(p) = \rho\, p^2/\kappa(p)^2$. In a viscoelastic medium with a completely 
monotonic relaxation modulus the function $Q$ is a CBF, hence by Theorem~\ref{thm:CBF2}  
$\kappa(p)^2/p = \rho\, p/Q(p)$ is a CBF.
\end{proof}

In the case of power law attenuation $\beta(p) = C\, p^\alpha$ the function   
$\kappa(p)^2/p = B^2\, p + 2 B\,C\, p^\alpha + C^2 \, p^{2 \alpha - 1}$
is a CBF if and only if $1/2 \leq \alpha \leq 1$. If $\alpha < 1/2$ then 
the exponent $2 \alpha - 1$ is negative and the function $\beta(p)^2/p$
is not monotone, hence it is not a BF; consequently the relaxation modulus
is not completely monotonic. The function $f(p) := p/(1 + p)$ is not a possible candidate for 
the wave number function because $p^2/f(p)^2$ is not a Bernstein function. Other 
counterexamples of this kind can be found in \cite{SerHan2010}.

We now examine the class of functions $\kappa(p)$ such that $\kappa(p)^2/p \in \CBF$.
Eq.~\eqref{eq:inteRepres} implies that $\kappa(p)^2/p \in \CBF$ if and only if 
$\beta(p) \in \CBF$ and $\beta(p)^2/p \in \CBF$.
We shall use the following notation:
$$\CBF^\alpha := \{ f^\alpha \mid f \in \CBF\}, \qquad \alpha \in \mathbb{R} $$
For any real function $f$ defined on $\overline{\mathbb{R}+}$ 
$$f\, \CBF := \{ f g \mid g \in \CBF\}$$
where $f g$ denotes the pointwise product of functions.
\begin{theorem}
$Q(p) \in \CBF$ if and only if $\kappa \in \CBF \cap p^{1/2}\, \CBF$.
\end{theorem}
\begin{proof}
If the function $Q(p) = p^2/\kappa(p)^2$ is a CBF then the function $Q(p)^{1/2}$ is also a CBF and 
the complex wave number function $\kappa(p) = \rho^{1/2}\, p/[Q(p)]^{1/2}$ 
is therefore also a CBF. On the other hand $\kappa(p) = \rho^{1/2}\,p^{1/2}\,[p/Q(p)]^{1/2}$. The third 
factor is the square root of a CBF and hence a CBF itself. Hence $\kappa \in p^{1/2}\, \CBF$
are also CBFs. We have thus proved the "only if" part.

For the converse we shall use the following identity (Theorem~7.11 in 
\cite{BernsteinFunctions}) for $\alpha \in [-1,1]$:
\begin{equation}\label{eq:CBFalpha}
\CBF^\alpha = \CBF \cap p^{\alpha-1}\,  \CBF 
\end{equation}
i.e.
$$\CBF^\alpha = \{ f \in \CBF \mid p^{1-\alpha}\, f(p) \in \CBF\}$$
Hence $p^{1/2}\, \CBF^{1/2} = p^{1/2}\,\CBF \cap \CBF$. 
By our hypothesis $\kappa \in p^{1/2}\, \CBF^{1/2}$, therefore $\kappa(p)^2 = p \,f(p)$,
where $f \in \CBF$. Thus
$Q(p) = p^2/\kappa(p)^2 = p/f(p) \in \CBF$. 
\end{proof}

We have thus proved that the mapping $Q \rightarrow \kappa$ is a bijective mapping 
of $\CBF$ onto the space $\mathcal{K} := \CBF \cap p^{1/2}\, \CBF$. 
This fact implies that the wave number function has another integral representation, viz.  
$\kappa(p) = p^{3/2} \, \int_{]0,\infty[} (p + r)^{-1} \, \lambda(\dd r)$ with $\lambda \in \mathfrak{M}$.

\section{Finite propagation speed.}

Since
$$\Re Q(p) = \int_{[0,\infty[} \frac{\vert p \vert^2 + r \,\Re p}{\vert p + r \vert^2} \, \mu(\dd r),$$
the function $\vert Q(p) \vert$ is a non-decreasing function of $\vert p \vert$ in the right half complex $p$-plane. 
It increases to infinity for $\vert p \vert \rightarrow \infty$ in the right half complex $p$-plane. If the spectral measure 
$\nu$ has infinite mass then 
the real part of the function $\kappa(p) = p/c_\infty + \beta(p)$ increases to infinity in the right half complex $p$-plane
(Theorem~\ref{thm:2finitewavefrontspeed}).

\begin{lemma} \label{lem:finitewsp}
Let $c_\infty < \infty$.

The function  $f(p) := [Q(p) \, \kappa(p)]^{-1}\, \e^{-\beta(p)\, r}$ is analytic in the right half plane and
tends to 0 for $\vert p \vert \rightarrow \infty$ uniformly with respect to $\arg p \in [-\upi/2,\upi/2]$.
\end{lemma}
\begin{proof}
By Theorem~\ref{thm:rebeta} 
$\vert f(p)\vert \leq \vert 1/[Q(p) \, \kappa(p)] \vert$ in $\mathbb{C}^+$.
If $\Re p \geq 0$ then
$$\Re Q(p) = \int_{[0,\infty[} \frac{\vert p \vert^2 + r \overline{p}}{\vert p + r\vert^2} \mu(\dd r) 
\geq \Re \vert p \vert^2 \int_{[0,\infty[} \frac{\mu(\dd r)}{\vert p + r\vert^2} \geq \vert p \vert^2 
\int_{[0,\infty[} \frac{\mu(\dd r)}{(\vert p\vert + r)^2}$$
in view of the inequality $\vert p + r\,\vert \leq r + \vert p \vert$.
Let $a > 0$ be sufficiently large so that $\mu([0,a]) > 0$. The integrand of the last integral is non-negative, hence 
$$\Re Q(p) \geq \vert p \vert^2 \int_{[0,a]} \frac{\mu(\dd r)}{(\vert p\vert + a)^2} = 
\frac{\vert p \vert^2}{(\vert p\vert + a)^2} \mu([0,a])$$
For an arbitrary positive  $\varepsilon_1 < \mu([0,a])$ a sufficiently large $R_1$ can be found so that  
$\Re Q(p) > \mu([0,a]) - \varepsilon_1$ 
and $\vert Q(p) \vert > \mu([0,a]) - \varepsilon_1$ for $\vert p \vert > R_1$.

In view of the inequality $1/\vert \kappa(p)\vert = 1/(\vert p\vert \, \vert B + \beta(p)/p\vert)$ Theorem~\ref{thm:betaineq} 
implies that for every positive $\varepsilon < 1$ there is a positive $R$ such that 
$1/\vert \kappa(p) \vert < [1/(1 - \varepsilon)]/ \vert p \vert$ for all $p$ with $\vert p \vert > R$ 
and $-\upi/2 \leq \arg p \leq \upi/2$.
\end{proof}

\begin{theorem}
If $c_\infty < \infty$, then the Green's function \eqref{eq:Green1D} vanishes for $t < x/c_\infty$.
\end{theorem}
\begin{proof}
Let $t < x/c_\infty$.

The Green's function is given by an expression of the form 
$u^{(1)}(t,x) = A \int_\mathcal{B} f(p) \e^{p\, (t - x/c_\infty)} \dd p$, where 
$\mathcal{B}$ is the Bromwich contour running from $\eta - \ii \infty$ to $\eta + \ii \infty$, $\eta > 0$, parallel to the imaginary axis, $A$ is a constant and the
function $f$ is defined in Lemma~\ref{lem:finitewsp}. Consider the complex contour $\mathcal{C}_r: \; p = r\, \e^{\ii \varphi}$,
with $\varphi$ running from $\upi/2$ to $-\upi/2$ and a fixed $r > 0$. Lemma~\ref{lem:finitewsp} and Jordan's lemma imply 
that the integral $A \int_{\mathcal{C}_r} f(p) \exp(p\, (t - x/c_\infty))\, \dd p$ tends to 0 as $r \rightarrow 0$. 
Let $\mathcal{B}_r$ be the straight line contour running from $\eta - \ii r$ to $\eta + \ii r$.
By the Cauchy 
theorem the integral of $f(p) \, \e^{p (t - x/c_\infty)}$ over the contour $\mathcal{B}_r + \mathcal{C}_r$ vanishes.
Hence, taking the limit $r \rightarrow \infty$, the integral over $\mathcal{B}$ vanishes.

Consequently $u^{(1)}(t,x) = 0$ for $t < x/c_\infty$.
\end{proof}

\begin{corollary}
If $c_\infty < \infty$ then $u^{(3)}(t,\x)$ vanishes for $t < \vert \x \vert/c_\infty$.
\end{corollary}
\begin{proof}
The thesis follows from eq.~\eqref{eq:1Dto3D}.
\end{proof}

\section{Dependence of wavefront smoothing on the spectral density.}

A frequent feature of wave propagation in real viscoelastic media is wavefront smoothing 
\cite{ReHrNo:VE,DeschGrimmer89b,HanQAM,HanMorro,SerHanCortona}. Using a terminology often adopted in mechanics, 
many linear viscoelastic media do not allow non-trivial discontinuity waves. It will now be shown
that absence or presence of non-trivial discontinuity waves depends on a property of the
dispersion-attenuation spectral measure. Wavefront smoothing is ultimately due to the singularity of
the derivative of the relaxation modulus at 0. We shall now relate it to the asymptotic properties of 
the attenuation-dispersion spectral measure.

If all the moments of the spectral density function $h$  
$$a_n := \int_0^\infty r^n \, h(r) \, \dd r$$
are finite, then 
\begin{equation} \label{eq:betaasymp1}
\beta(p) = p \int_0^\infty \frac{h(r)\, \dd r}{p + r} \sim_\infty \sum_{n=0}^\infty (-1)^n \,a_n \, p^{-n} 
\end{equation}
\cite{Wong,McClureWong78}. The dominating term is a positive constant $a_0$.
A special case is a finite bandwidth spectral density such as $h(r) = K \,\chi_{[a,b]}(r)$, $K > 0$, $0 \leq a < b < \infty$.
In this case 
$\beta(p) = K \, \ln[(p + b)/(p + a)] \sim_\infty K\, (b - a) + \OO\left[p^{-1}\right]$.

If the function $h$ decays at an algebraic rate then some higher order moments are infinite and \eqref{eq:betaasymp1} 
does not hold. Assuming the asymptotic expansion of the spectral density 
\begin{equation}
h(r) \sim_\infty \sum_{n=0}^\infty b_n \, r^{- n - \alpha}, \qquad 0 < \alpha < 1
\end{equation} 
the dissipation-attenuation function has the following asymptotic expansion at infinity \cite{McClureWong78}
\begin{equation}
\beta(p) \sim_\infty \frac{\upi}{\sin(\upi \alpha)}  \sum_{n=0}^{N-1} (-1)^n\, \left[b_n \,p^{1-n-\alpha} 
- n \,c_n \, p^{-n} \right] + R_N
\end{equation}
where
$$c_n := \int_0^\infty f_n(r)\, r^{n - 1} \, \dd t, \quad f_n(r) := h(r) - \sum_{k=0}^n b_n \, r^{-k-\alpha}$$
$$R_N := \frac{(-1)^N}{p^{N-1}} \int_0^\infty \frac{r^n \, f_n(r)}{p + r} \dd r$$
The asymptotic expansion of $\beta$ is valid in the entire cut complex plane, $\vert \arg p \vert < \upi$
\cite{McClureWong78}.  
The first term of the expansion of the dissipation-attenuation function $\beta$ is now $b_0 \, p^{1 - \alpha}$
while the attenuation function 
$\mathcal{A}(\omega) := \Re \beta(-\ii \omega) \sim_\infty b_0 \, \sin(\upi \alpha) \,\omega^{1-\alpha}$.  

Consequently, if the spectral density $h$ decays algebraically at infinity, then 
$\vert \exp( -\ii \omega (t - x/c_\infty) - \beta(-\ii \omega)\, \vert x \vert)\vert  \leq \vert  \exp(-\mathcal{A}(\omega) \,\vert x \vert) \vert$
vanishes asymptotically like $\exp\left(- b_0 \, \sin(\upi \alpha) \,\omega^{1-\alpha} \, r\right)$. Hence the integral 
\eqref{eq:Green1D} is absolutely convergent and therefore the function $u^{(1)}$ is continuous. Furthermore, 
the derivatives of $\D^n_t\, 
\D^m_x\, u^{(1)}(t,x)$ are also given by absolutely convergent integrals. Consequently $u^{(1)}(t,x)$ is a smooth
function of both arguments in $\mathbb{R}_+ \times \mathbb{R}_+$. In particular it has continuous derivatives of
arbitrary order at the wavefront $t = x/c_\infty$. Thus the wavefront does not carry any discontinuity of the Green's function
nor its derivatives. Since the Green's function vanishes for $t < x/c_\infty$, it gradually decays to 0 with 
all its derivatives.

Eq.~\eqref{eq:1Dto3D} implies that $u^{(3)}(t,\x)$ is also a smooth function of $(t,\x) \in \mathbb{R_+} \times 
(\mathbb{R}^3\setminus \{ 0 \})$. 

The case of a strongly singular relaxation modulus is analyzed in Theorem~\ref{thm:strongsing}. 
\begin{theorem} \label{thm:strongsing}
If $G(t) = \int_{[0,\infty[} \e^{- t r} \, \mu(\dd r)$ with $\mu \in \mathfrak{M}$ and 
$\mu([0,r]) \sim_\infty r^\alpha \, l(r)$, where $0 < \alpha < 1$ and $l$ is slowly varying at infinity, 
then $G_0 = \infty$, $c_\infty = \infty$, $Q(p) \sim_\infty c_\alpha\, p^\alpha \, l(p)$,
$\kappa(p) \sim_\infty \rho^{1/2} \, c_\alpha^{\;-1/2} \,p^\gamma \, l(p)^{-1/2}$, $1/2 < \gamma < 1$  and
$\nu([0,r]) \sim_\infty \rho^{1/2} \, c_\alpha^{\;-1/2} \,r^\gamma \, l(p)^{-1/2}$,
where $\gamma := 1 - \alpha/2$, $c_\alpha := \upi\alpha/\sin(\upi\alpha)$. 
\end{theorem}
If $\mu([0,r]) \sim_\infty r^\alpha \, l(r)$, $0 < \alpha < 1$, then, by the Karamata Abelian Theorem
(Theorem~\ref{thm:KaramataAbelianInfinity}),
$G(t) \sim_0 t^{-\alpha} \, l(1/t)/\Gamma(1-\alpha)$. In this case the Green's function is a complex analytic 
function of $(t,x)$ in a neighborhood of $\mathbb{R}_+ \times \mathbb{R}$ \cite{ReHrNo:VE}.
(This follows from a more general fact that the Green's function is analytic except at the wavefront; in the strongly 
singular case the Green's function does not have a wavefront). Analyticity implies that the Green's function cannot
vanish on any open subset of $\mathbb{R}_+ \times \mathbb{R}$ unless it is identically zero.  

The case of weakly singular relaxation modulus is somewhat more complicated. The relaxation modulus is weakly singular 
if $G$ is bounded and $G^\prime(t) \sim -b \, t^{-\alpha}\, l(t)$, $b > 0$, $0 < \alpha < 1$. In this case
$l(r) := \mu(]0,r[) \sim_\infty \mu(]0,\infty[) - a \, r^{-\alpha}$, with $a = b\, \Gamma(1-\alpha)$,
 $G_0 = \mu(\{ 0 \}) < \infty$, $G_0 - G_\infty = \mu(]0,\infty[)$. $l$ is a function of slow variation at infinity
vanishing at 0. By Valiron's Theorem 
$$Q(p) = G_\infty + p \int_{]0,\infty[} \frac{\mu(\dd r)}{p +  r}\sim_\infty G_\infty + \mu(]0,\infty[) - a \, p^{-\alpha}
= G_0 - a \, p^{-\alpha}$$
Hence
$$
\kappa(p) \sim_\infty \left(\frac{\rho}{G_0}\right)^{1/2} \frac{p}{\left( 1 - a p^{-\alpha}/G_0\right)^{1/2}} 
$$
This implies that $\beta(p) \sim_\infty [a/(2 G_0 \, c_\infty)]\, p^\gamma$ 
where $\gamma := 1 -\alpha \in\; ]0,1[$, and $$\nu([0,r]) \sim_\infty 
\frac{a}{2 G_0\, c_\infty} \frac{\sin(\upi \alpha)}{\upi (1 - \alpha)} r^{1 - \alpha}$$ The dissipation-attenuation spectral density has an 
infinite bandwidth and an algebraic decay at infinity and the Green's functions are infinitely smooth at 
the wavefront $x = c_\infty\, t$.

Summarizing, 
\begin{itemize}[(i)]
\item if $\nu(\dd r) = h(r) \, \dd r$ and all the moments of the spectral density $h$ are finite, then the 
Green's function can have discontinuities
at the wavefront;
\item if $\nu(\dd r) = h(r) \, \dd r$ and the spectral density function $h$ decays at an algebraic rate then 
the wavefront does not carry any discontinuity of the Green's function nor its derivatives of arbitrary order.
\end{itemize}

\section{Attenuation and dispersion functions.}

Define the attenuation function $\mathcal{A}$ and the dispersion function
$\mathcal{D}$ by the equations
\begin{gather}
\mathcal{A}(\omega) := \Re \beta(-\ii \omega) \equiv \Re \kappa(-\ii \omega) = \omega^2 \int_{]0,\infty[\;} 
\frac{\nu(\dd r)}{\omega^2 + r^2} \label{eq:frAtt}\\
\mathcal{D}(\omega) := -\Im \beta(-\ii \omega) = \omega \int_{]0,\infty[\;} 
\frac{r\, \nu(\dd r)}{\omega^2 + r^2} \label{eq:frDisp}
\end{gather}

By Theorem~\ref{thm:betaineq} 
\begin{equation} \label{eq:attdispassymp}
\mathcal{A}(\omega) = \oo_\infty[\omega], \qquad \mathcal{D}(\omega) = \oo_\infty[\omega] 
\end{equation}

Note that $\mathcal{A}(\omega) \geq 0$ and $\sgn\, \mathcal{D}(\omega) = \sgn\, \omega$. 
Recalling the Green's function \eqref{eq:Green1D} and the definition of
$\kappa(p)$, this implies an outgoing sense of propagation and a non-negative attenuation along each
radial direction. 

Since $x/(1 + x)$ is an increasing function, eq.~\eqref{eq:frAtt} implies that the
attenuation function $\mathcal{A}$ is non-decreasing.

The phase speed 
$c(\omega) := -\omega/\Im \kappa(-\ii \omega) = 1/(1/c_\infty - \Im \beta(-\ii \omega)/\omega)$ 
is related to the dispersion function by the equation 
\begin{equation} \label{eq:c0}
\frac{1}{c(\omega)} = \frac{1}{c_\infty} + \frac{\mathcal{D}(\omega)}{\omega}
\end{equation}
Note that $c(\omega) \leq c_\infty$ because $\mathcal{D}(\omega) \geq 0$. Moreover
$\mathcal{D}(\omega)/\omega = \int_{]0,\infty[} r \, \nu(\dd r)/(\omega^2 + r^2)$ is
a non-increasing function of $\omega$. Therefore the phase speed $c(\omega)$ is a non-decreasing
function of frequency.

Eq.~\eqref{eq:c0} and eq.~\eqref{eq:attdispassymp} imply that 
\begin{equation} \label{eq:phasespeedlimit}
\lim_{\omega \rightarrow \infty} c(\omega) = c_\infty
\end{equation}

If $Q(0) = G_\infty > 0$ then, in view of eq.~\eqref{eq:kp}, 
$\lim_{p\rightarrow 0} \kappa(p)/p = (\rho/G_\infty)^{1/2}$ for all $p \in \mathbb{C}$. 
The inverse of this limit will be denoted by the symbol $c_0$. 
In particular, for $p = -\ii \omega$, $\omega \in \mathbb{R}$,  this proves the following theorem:
\begin{theorem} \label{thm:yyy} If $G_\infty > 0$ then 
$$\lim_{\omega\rightarrow 0} \mathcal{A}(\omega)/\omega = 0; \quad D := \lim_{\omega\rightarrow 0} 
\mathcal{D}(\omega)/\omega < \infty$$
and $1/c_0 = \lim_{\omega\rightarrow 0} [1/c(\omega)] = 1/c_\infty + D$. 
\end{theorem} 
Note that $D \geq 0$ and $D > 0$ unless $G(t) = \const$. Under the hypotheses of
Theorem~\ref{thm:yyy} the phase speed satisfies the inequality $c_0 < c(\omega) < c_\infty$, while
$0 < c(\omega) < c_\infty$ in the opposite case.

\begin{corollary}
$G_\infty > 0$ entails that
\begin{equation} \label{eq:xx}
D = \int_{[0,\infty[} \frac{\nu(\dd r)}{r} < \infty
\end{equation}
\end{corollary}
\begin{proof}
Indeed, 
$$\frac{\mathcal{D}(\omega)}{\omega} = \int_{[0,\infty[}  \frac{\nu(\dd r)}{r} \frac{1}{1 + \omega^2/r^2}$$
For $\omega \rightarrow 0$ the 
function $1/(1 + \omega^2/r^2)$ increases monotonically to 1. Therefore, if it is assumed that the integral 
in \eqref{eq:xx} is infinite then, by the Fatou lemma, $\mathcal{D}(\omega)/\omega$ tends to infinity
for $\omega \rightarrow 0$. On the other hand, if the inequality in \eqref{eq:xx} is satisfied,
then 
\begin{equation} \label{eq:ey}
D = \lim_{\omega\rightarrow 0} \frac{\mathcal{D}(\omega)}{\omega} = \int_{[0,\infty[} \frac{\nu(\dd r)}{r}
\end{equation} by 
the Lebesgue Dominated Convergence Theorem.
\end{proof}

\section{High-frequency behavior and the Kramers-Kronig dispersion relations.}
\label{sec:KK}

The Kramers-Kronig (K-K) dispersion relations are the Sochocki-Plemelj formulae 
 following from the fact the dissipation-attenuation function 
$\beta(p)$ is the Laplace transform of a causal distribution $F(t)$. Indeed,  
\begin{equation} \label{eq:KK1}
\frac{\beta(p)}{p} = \int_{[0,\infty[} \frac{\nu(\dd r)}{p + r} = L^2(\nu)(p)
\end{equation}
where $L$ denotes the Laplace transformation. The integral representation of the function $\beta$ 
shows that all the singularities of the complex analytic function $\beta(p)/p$ lie in the closure of the left
half $p$-plane. By Jordan's lemma and Theorem~\ref{thm:betaineq} the source function 
$g$ of $\beta_1(p) := \beta(p)/p$ vanishes for $t < 0$. Applying the inverse Laplace 
transformation $L^{-1}$ to both sides of eq.~\eqref{eq:KK1} we have
$$g(t) = L^{-1}(\beta_1)(t) = L(\nu) = \int_{[0,\infty[} \e^{-r t}\, \nu(\dd r)$$
for $t > 0$ with $\nu \in \mathfrak{M}$. By Theorem~\ref{thm:locint} the function $g$ is LICM. Hence 
its primitive 
$$f(t) := \int_0^t g(s)\, \dd s \equiv \int_{]0,\infty[} \left[1 - \e^{-t r}\right] \, \nu(\dd r)$$ 
is a Bernstein function continuous over $[0,\infty[$ and $f(0) = 0$. 
It follows that the dissipation-attenuation function $\beta(p)$ is the Laplace transform of the   
causal distribution $D^2 \, f$ of second order. For $\beta(p) = C\, p^\alpha$, with 
$0 < \alpha < 1$, $C > 0$, the function $f$ is $C\,t_+^{\;1-\alpha}/\Gamma(2-\alpha)$.

\begin{theorem} \label{thm:highfreq}
If the $\nu$ measure of the closed segment $[0,r]$ is a function of $r$ regularly varying with index $1 - \alpha$,
$$\nu([0,r]) \sim_\infty r^{1-\alpha}\, l(r)$$
where $l$ is a function of slow variation at infinity, then $\alpha > 0$ and 
\begin{gather}
\mathcal{A}(\omega) \sim_\infty \frac{(1-\alpha)\, \upi}{2 \cos(\alpha \pi/2)} \omega^{1-\alpha}\, l(\omega)\\
\mathcal{D}(\omega) \sim_\infty \frac{(1-\alpha)\, \upi}{2 \sin(\alpha \pi/2)} \omega^{1-\alpha}\, l(\omega)
\end{gather}
\end{theorem}
\begin{proof}
Since $\nu \in \mathfrak{M}$, eq.~\eqref{eq:doss} implies that $\alpha > 0$. 

Let $F(\omega) := \int_{[0,\infty[} (\omega^2 + r^2)^{-1}\,\nu(\dd r)
\equiv \int_{[0,\infty[} (\omega^2 + r^2)^{-1}\, \dd \nu([0,r])$, where the integral on the extreme right-hand side is a 
Stieltjes integral. 
The function $F$ can be expressed in terms of a Stieltjes transform by changing the integration variable.
In terms of a new measure $\mu([0,s]) := \nu([0,\smash{\sqrt{s}]}) = s^{(1 - \alpha)/2}\, l\left(\sqrt{s}\right)$, 
\begin{equation} \label{eq:attenAsym}
F(\omega) = \int_{[0,\infty[} \frac{\mu(\dd s)}{\omega^2 + s}
\end{equation}
By Valiron's theorem (Theorem~\ref{thm:Valiron}) 
$$F(\omega) = \frac{(1-\alpha)\, \upi}{2 \cos(\upi \alpha/2)} \omega^{1-\alpha} l(\omega)$$
This proves the theorem.
\end{proof} 

The function $F$ in eq.~\eqref{eq:attenAsym} is differentiable for $\omega > 0$, hence the attenuation 
and dispersion 
functions $\mathcal{A}$ and $\mathcal{D}$ are differentiable. The function $\omega^{-1}\,\mathcal{A}(\omega)
\sim_\infty \omega^{-\alpha}$ by Theorem~\ref{thm:highfreq}, hence it belongs to the space 
$\mathcal{D}_L^{\prime(1)}$ (eq.~(1.8.12) in \cite{Nussenzveig} Sec.~1.7--1.8).
Consequently under the above hypotheses 
the functions $\mathcal{A}(\omega) = \Re \beta(-\ii \omega)$ and $\mathcal{D}(\omega) = 
-\Im \beta(-\ii \omega)$ satisfy the K-K dispersion relations with one subtraction:
\begin{equation} \label{eq:KK}
\mathcal{A}(\omega) - \mathcal{A}(\omega_0) = -\frac{(\omega - \omega_0)}{\upi} \mathrm{vp} \int_{-\infty}^\infty 
\frac{\mathcal{D}(\omega^\prime) - \mathcal{D}(\omega_0)}{(\omega^\prime - \omega_0)\, (\omega^\prime-\omega)} 
\dd \omega^\prime
\end{equation}
where "vp" indicates that the integral is to be taken in the sense of principal value. 

\section{Attenuation and dispersion functions. Low-frequency behavior.}

Low-frequency behavior of the attenuation function provides a useful test whether a specimen of the material subjected to 
constant strain for $t > 0$ relaxes to zero stress - that is, whether $G_\infty = 0$. Materials with vanishing 
$G_\infty$ are known as viscoelastic fluids. Scalar models considered in this paper represent either 
longitudinal or shear waves. It turns out that the same material can behave under tension or compression 
like a viscoelastic solid and under shear strain as a viscoelastic fluid.

If $\nu \in \mathfrak{M}$ and the $\nu$ measure $\nu([0,r])$ of the closed segment $[0,r]$ is a function of $r$ 
regularly varying at 0 with index $\gamma$, then $\gamma > 0$.
\begin{theorem} \label{thm:lowfreq}
If the Radon measure $\nu \in \mathfrak{M}$ is regularly varying at 0 with $\nu([0,r]) = r^\gamma\, l(r)$,
where $l(r)$ is slowly varying at 0 and  $0 < \gamma < 1$, then
\begin{enumerate}[(i)]
\item $\mathcal{A}$ is regularly varying at 0 with $$\mathcal{A}(\omega) = \frac{\gamma \upi}{2 \sin(\gamma \upi/2)} 
\omega^\gamma\, l(\omega^2)$$ with $l$ slowly varying at 0;
\item
$\mathcal{D}$ is regularly varying at 0 with
$$\mathcal{D}(\omega) \sim_0 \frac{\gamma \upi}{\cos(\gamma \upi/2)} \omega^\gamma\, l(\omega^2);$$
\item the Q factor is asymptotically constant at 0,
$\mathcal{Q}(\omega) \sim_0 4 \upi \cot(\gamma \upi/2)$.
\end{enumerate}
\end{theorem}
\begin{proof}
\noindent Ad (i) 
$$\mathcal{A}(\omega) = \omega^2 \int_{[0,\infty[} \frac{\nu(\dd r)}{\omega^2 + r^2} =
\omega^2 \int_{[0,\infty[} \frac{\mu(\dd s)}{s + \omega^2}$$
where the integration variable has been changed to $s = r^2$ and 
$\mu([0,s]) \sim_0 s^{\gamma/2}\, l(1/s)$. For example, if $\nu(\dd r) = h(r)\, \dd r$ with
$h(r) \sim_0 \gamma r^{\gamma-1}\, l(r)$, then
$\mu(\dd s) = g(s) \, \dd s$ with 
$g(s) = h\left(\smash{\sqrt{s}}\right)/(2 \sqrt{s}) \sim_0 (\gamma/2) s^{\gamma/2 - 1} l(\sqrt{s})$.
The function $l(\smash{\sqrt{s}})$ is slowly varying at 0, hence 
Lemma~\ref{thm:Valiron} implies the thesis.

\noindent Ad (ii) 
$$\mathcal{D}(\omega) \sim_0 \omega \int_{[0,\infty[} \frac{r \nu(\dd r)}{r^2 + \omega^2} =
\omega \int_{[0,\infty[} \frac{\sqrt{s}\, \mu(\dd s)}{s + \omega^2}$$
Valiron's Theorem  implies the thesis.
\end{proof}

Theorem~\ref{thm:lowfreq} is applicable to viscoelastic fluids. In the case of viscoelastic solids 
eq.~\eqref{eq:xx} implies that $\gamma > 1$ and the hypotheses of 
Theorem~\ref{thm:lowfreq} are not satisfied. In this case Theorem~\ref{thm:yyy} implies 
that $\mathcal{D}(\omega) = D \, \omega + \oo_0[\omega]$. If $\mathcal{A}(\omega)$ 
is regularly varying at 0 then $\mathcal{A}(\omega) \sim_0 \vert \omega\vert ^{\alpha+1}\, l(\omega)$,
where $l$ is slowly varying at 0 and either $\alpha > 0$ or $l(0) = 0$, for example
$\mathcal{A}(\omega) \sim_0 C\, \vert \omega\vert\, \ln^\beta(1 + \omega)$ with $C, \beta > 0$.

Attenuation of longitudinal waves in polymers and bio-tissues exhibits the power law behavior  
 $\mathcal{A}(\omega) \sim_0 C \vert \omega\vert^{1 + \alpha}$. The asymptotic behavior described by (i) 
of Theorem~\ref{thm:lowfreq} is observed for shear wave attenuation in minerals (e.g. \cite{JacksonAl2004}).

\section{A few examples of analytic viscoelastic models.}
\label{sec:examples}

We shall now demonstrate an application of the spectral method to the numerical determination of 
attenuation and
dispersion in viscoelastic media defined by the most popular analytic expressions for the complex modulus.

\subsection{Power law.}

If $\nu([0,r]) = a \,r^\gamma$, $0 < \gamma < 1$, then $\mathcal{A}(\omega) = A \, \omega^\gamma$,
where $$ A = a\, \gamma\, \int_0^\infty \frac{y^{\gamma-1} \dd y}{1 + y^2}.$$

\subsection{Finite bandwidth.}
\label{sec:fb}

If $\nu(\dd r) = C\, \chi_{[a,b]}(r) \, \dd r$, with $0 < a < b < \infty$, $C > 0$, then the 
attenuation function
$$\mathcal{A}(\omega) = C\, \omega \, \left[\tan^{-1}(b/\omega) - \tan^{-1}(a/\omega)  \right]$$
is asymptotically constant at infinity and $\OO\left[\omega^2\right]$ at 0:
\begin{equation}
\mathcal{A}(\omega) \begin{cases} \sim_0 C\, \omega^2\, (1/a - 1/b)\\
\sim_\infty C\, (b - a)
\end{cases}
\end{equation}
On the other hand
$1/c_0 = \lim_{p\rightarrow 0} \kappa(p)/p = 1/c_\infty + C\, \ln(b/a)$,
hence $\mathcal{D}(\omega) \sim_0 C\, \ln(b/a)\, \omega$.

\subsection{The Cole-Cole relaxation model.}
\label{sec:ColeCole}

The Cole-Cole model is defined by the Laplace transform of the relaxation modulus of the form
\begin{equation} \label{eq:CC}
\tilde{G}^{\mathrm{CC}}_\alpha(p) = \frac{G_\infty}{p} \frac{1 + a \,(\tau p)^\alpha}{1 + (\tau p)^\alpha}, 
\qquad  a > 1, \quad 0 < \alpha < 1, \quad G_\infty > 0, \quad \tau > 0
\end{equation}
\cite{ColeCole,BagleyTorvik3}. The Cole-Cole model \cite{HanColeJCA} and its generalizations \cite{RossikhinShitikova1} are 
equivalent to linear viscoelastic models based on fractional derivatives considered in 
\cite{MainardiVE}. 
The function $G^{\mathrm{CC}}_\alpha$ is bounded completely monotonic (Appendix~\ref{app:HNCM}). Using the formula 
\begin{equation}  
\int_0^t \e^{-p t}\,E_{\alpha,\beta}(-\lambda\,t^\alpha) \, \dd t = \frac{p^{\beta-1}}{p^\alpha + \lambda}
\end{equation}
\cite{PodlubnyBook} the relaxation modulus $G^{\mathrm{CC}}_\alpha$ can be expressed in terms of the Mittag-Leffler 
function $E_\alpha = E_{\alpha,1}$ 
\begin{equation}
G^{\mathrm{CC}}_\alpha(t) = G_\infty \, \theta(t) \, \left[ 1 + (a - 1) \, 
E_\alpha\left(-t^\alpha/b\right)\right]
\end{equation}
The relaxation modulus $G^{\mathrm{CC}}_\alpha$ has a finite limit $G_0 = G_\infty\, a$ at $t = 0$. 

The complex wave number function of the Cole-Cole model is given by the formula
\begin{equation}
\kappa^{\mathrm{CC}}_\alpha(p) = \frac{p}{c_0} \, \left[ \frac{1 + (\tau p)^\alpha}{1 + a (\tau p)^\alpha}\right]^{1/2} =
B \, p + p \int_{]0,\infty[} \frac{\nu(\dd r)}{p + r} 
\end{equation}
An explicit integral expression for the attenuation and dispersion functions of the Cole-Cole model 
can be derived from
eqs~(\ref{eq:inteRepres}--\ref{eq:inteRepres1}) and \eqref{eq:jumpmeasure}. In the first place 
Corollary~\ref{cor:B} implies that 
\begin{equation}
B = \lim_{p\rightarrow\infty} \frac{\kappa^{\mathrm{CC}}_\alpha(p)}{p} =  a^{-1/2}/c_0
\end{equation}
The wavefront speed equals $c_\infty = 1/B$.

The spectral density $h^{\mathrm{CC}}_\alpha$ of the Cole-Cole model is given by the formula
$$h^{\mathrm{CC}}_\alpha(r) = \frac{1}{\upi \, c_0} \Im \left[\frac{\kappa^{\mathrm{CC}}_\alpha(r \exp(\ii \upi))}{r \exp(\ii \upi)} - 
\left(\frac{1}{a}\right)^{1/2} \right]  = \frac{1}{\upi\, c_0} \Im \left[ \frac{1 + (\tau r)^\alpha \, 
\exp(\ii \alpha \upi)}{1 + a\, (\tau r)^\alpha \, \exp(\ii \alpha \upi)}  \right]^{1/2}$$
Let
$$\mathcal{J} := \Im \frac{1 + (\tau p)^\alpha}{1 + a\, (\tau p)^\alpha}, \quad  \mathcal{R} := 
\Re \frac{1 + (\tau p)^\alpha}{1 + a\,(\tau p)^\alpha}$$ where $p = r\, \exp(\ii \upi)$. 
The imaginary part $Y$ of $\left[\left( 1 + (\tau p)^\alpha \right)/\left( 1 + a\,(\tau p)^\alpha \right)\right]^{1/2}$ 
is a solution
of the bi-quadratic equation $Y^4 + \mathcal{R}\, Y^2 - \mathcal{J}^2/4 = 0$ and is non-negative, hence
$Y = \sqrt{\sqrt{\mathcal{R}^2 + \mathcal{J}^2} - \mathcal{R}}/\sqrt{2}$. Now
$\mathcal{J} = \mathcal{J}_1/Z^2$, $\mathcal{R} = \mathcal{R}_1/Z^2$, where 
$\mathcal{J}_1 = - (a - 1)\,\sin(\upi \alpha)\, (\tau r)^\alpha$,
$\mathcal{R}_1 = 1 + a\,(\tau r)^{2 \alpha} + (a + 1)\, \cos(\upi \alpha)\, (\tau r)^\alpha$ and
$Z = \sqrt{1 + a^2 \, (\tau r)^{2 \alpha} + 2 a \cos(\upi \alpha) \,(\tau r)^\alpha}$. Hence the dispersion-attenuation  
spectral measure $\nu$ of the Cole-Cole model has a density
\begin{equation}
h^{\mathrm{CC}}_\alpha(r) =  \frac{1}{\upi\,c_0\, \sqrt{2}} \frac{\sqrt{\vert \mathcal{R}_1\vert\,\left[
 \sqrt{1 + (\mathcal{J}_1/\mathcal{R}_1)^2}  - \sgn(\mathcal{R}_1)\right]}}{Z}
\end{equation}
Note that
\begin{equation}
h^{\mathrm{CC}}_\alpha(r) \begin{cases}
\sim_\infty \frac{b^{1/2}}{\upi\,c_\infty} \frac{a - 1}{a} \sin(\alpha \upi) \, (\tau r)^{-\alpha}\\
\sim_0 \frac{1}{\sqrt{2} \upi c_0} (a - 1)\, \sin(\alpha \upi) \, (\tau r)^\alpha
\end{cases}
\end{equation}
and the function $h^{\mathrm{CC}}_\alpha(r)/r$ is integrable. Thus $\mathcal{D}^{\mathrm{CC}}_\alpha(\omega)
\sim_0 D\, \omega$. The asymptotic properties of the attenuation function follow from Theorem~\ref{thm:lowfreq} 
\begin{equation}
\mathcal{A}^{\mathrm{CC}}_\alpha(\omega)  \begin{cases}
\sim_\infty \frac{1}{2} \frac{a - 1}{a\, c_0} \sin(\alpha \upi/2) \, (\tau \omega)^{1 - \alpha}\\
\sim_0 \frac{1}{\sqrt{2}\, c_0} (a - 1) \, \sin(\alpha \upi/2) \, (\tau \omega)^{1 + \alpha}
\end{cases}
\end{equation}

For $\alpha \rightarrow 1$ the dispersion-attenuation spectral density $h^{\mathrm{CC}}_\alpha$ tends to the 
spectral density of the Standard Linear Solid
\begin{equation}
h^{\mathrm{SLS}}(r) = \begin{cases} 0, & \mathcal{R}_1 > 0 \\
\frac{1}{\upi\, c_0} \frac{\sqrt{-\mathcal{R}_1}}{\vert a\, r - 1\vert }, & \mathcal{R}_1 < 0
\end{cases}
\end{equation}
or, equivalently
\begin{equation}
h^{\mathrm{SLS}}(r) = \frac{1}{\upi\, c_0} \sqrt{\frac{1 - \tau\, r}{a\, \tau \,r - 1}} \, \chi_{[1/a,1]}(\tau \,r)
\end{equation}
In contrast to the Cole-Cole model, the spectrum of the SLS has a finite bandwidth and 
$\int_0^\infty h^{\mathrm{SLS}}(r) \, \dd r < \infty$. 

The attenuation functions of the Cole-Cole and the SLS are now given by
the explicit expressions
\begin{equation} \label{eq:attnCC}
\mathcal{A}^{\mathrm{CC}}_\alpha(\omega) = \omega^2 \int_0^\infty \frac{h^{\mathrm{CC}}_\alpha(r)}{r^2 + \omega^2} \, \dd r
\end{equation}
for $0 < \alpha < 1$ and $\alpha = 1$, respectively. 
The dispersion functions are given by
\begin{equation} \label{eq:dispCC}
\mathcal{D}^{\mathrm{CC}}_\alpha(\omega) = \omega \int_0^\infty \frac{r\, h^{\mathrm{CC}}_\alpha(r)}{r^2 + \omega^2} \, \dd r
\end{equation}
$0 < \alpha \leq 1$. 
The attenuation function and the phase speed $c(\omega)$ with
$c_\infty = 5000 \,\mathrm{m}/\mathrm{s}$ are shown in Fig.~\ref{fig:1}.

\begin{figure}
\begin{minipage}[t]{0.49\linewidth} 
\includegraphics[width=\textwidth]{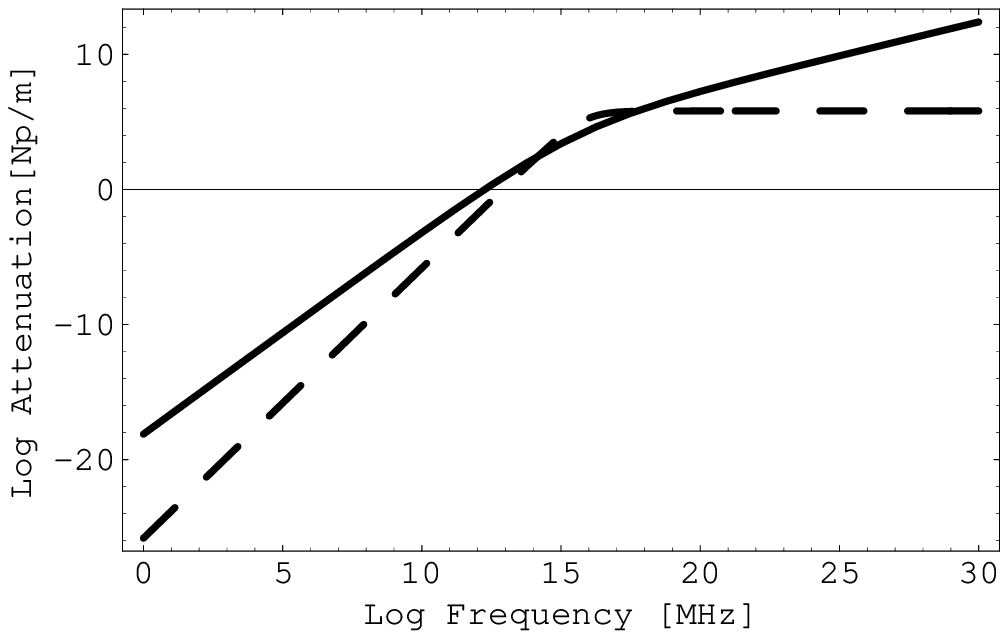}
\begin{center} {\footnotesize (a) Attenuation.} \end{center}
\end{minipage} \hspace{0.5cm}
\begin{minipage}[t]{0.49\linewidth}
\includegraphics[width=\textwidth]{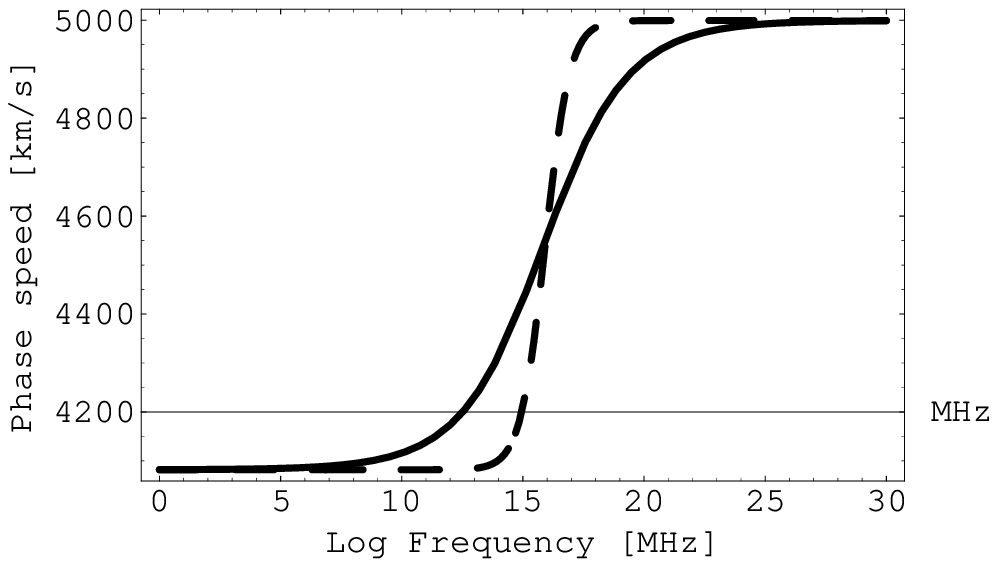}
\begin{center} {\footnotesize (b) Phase speed.} \end{center}
\end{minipage}
\caption{Attenuation function and phase speed plotted vs logarithm of angular frequency in MHz 
for the SLS and Cole-Cole models, $a = 1.5$, $c_\infty = 5 \, \mathrm{km/s}$, $\tau = 10^{-13} \, \mathrm{s}$.} \label{fig:1}
{\footnotesize Solid line: Cole-Cole, $\alpha=1/2$; dashed line: Standard Linear Solid.}
\end{figure}

Equations~\eqref{eq:attnCC} and \eqref{eq:dispCC} are convenient for numerical computation of the
attenuation and dispersion functions. 

\subsection{The Havriliak-Negami relaxation.}

The Havriliak-Negami relaxation model is defined by the Laplace transform of the relaxation modulus of the form
\begin{multline} \label{eq:HN}
\tilde{G}^{\mathrm{HN}}_{\alpha,\gamma}(p) = \frac{G_0}{p} \left[ 1 - \frac{b}{\left(1 + (\tau\, p)^\alpha\right)^\gamma}\right], \\
0 < b \leq 1, \quad 0 < \alpha < 1, \quad 0 < \gamma \leq 1, \quad G_0 > 0 
\end{multline}
\cite{HavriliakHavriliak}. The corresponding relaxation modulus is given by a 4-parameter formula 
involving the Prabhakar Mittag-Leffler function $E^\gamma_{\alpha,\alpha\gamma}$ 
\cite{HanSerDielectrics07}. The Havriliak-Negami relaxation modulus $G^{\mathrm{HN}}_{\alpha,\gamma}$ is a 
CM function (Appendix~\ref{app:HNCM}).

For $\gamma = 1$ the Havriliak-Negami relaxation modulus reduces to the Cole-Cole model with $a = 1/(1 - b)$.

Eq.~\eqref{eq:kp} implies that $c_\infty = (G_0/\rho)^{1/2}$ and
$\beta(p)/p = \left[1 - b/\left(1 + (\tau p)^\alpha \right)^\gamma\right]^{-1/2} - 1/c_\infty$. 
Hence the spectral density is given by the expression 
$$h^{\mathrm{HN}}_{\alpha,\gamma}(r) = \frac{1}{\upi c_\infty} \Im Z^{-1/2}$$
where $Z := 1 - b/Y$ and $Y := \left( 1 + (\tau r)^\alpha \, \exp(\ii \upi \alpha) \right)^\gamma$.

Substitution of the expressions  
$$ \vert Y \vert = g(r) := \left( 1 + 2 (\tau r)^\alpha\, \cos(\upi \alpha) + (\tau r)^{2 \alpha}\right)^{\gamma/2}$$
$$\Re Y = g(r)\, \cos(\gamma f(r)), \quad \Im Y = g(r)\, \sin(\gamma f(r))$$
where
$$f(r) := \tan^{-1}\left((\tau r)^\alpha\, \sin(\upi \alpha)/\left( 1 + (\tau r)^\alpha \, \cos(\upi \alpha)\right)\right)$$
in the formulae
$$\Re Z = 1 - b \frac{\Re Y}{\vert Y\vert^2}, \quad \Im Z = b \frac{\Im Y}{\vert Y\vert^2}$$
yields the formula 
$$\vert Z\vert = k(r) := \sqrt{1 + b^2/g(r)^2 - 2 b\, \cos(\gamma f(r))/g(r)}$$
Hence, using the identity 
$$\Im Z^{-1/2} = - \frac{\Im Z^{1/2}}{\vert Z\vert} = \mp \frac{\sqrt{\vert Z\vert - \Re Z}}{\sqrt{2}\, 
\vert Z\vert}$$ and noting that the lower sign applies in the problem at hand, the final result is 
\begin{equation} \label{eq:hHN}
h^{\mathrm{HN}}_{\alpha,\gamma}(r) = \frac{1}{\upi \sqrt{2} c_\infty} \frac{\sqrt{k(r) - 1 + b \cos(\gamma  f(r))/g(r)}}{k(r)}
\end{equation}

Note that $f(0) = 0$ and therefore $k(0) = 1 - b/g(0)$. This implies that  $h^{\mathrm{HN}}_{\alpha,\gamma}(0) = 0$ 
and $D = \int_0^\infty [h^{\mathrm{HN}}_{\alpha,\gamma}(r)/r] \dd r < \infty$.

In comparison with the Cole-Cole model the parameter $\gamma$ in the Havriliak-Negami model allows 
decoupling the high-frequency asymptotics of the wave 
number function from its low-frequency asymptotics. 
The high-frequency asymptotics  $Q(p) \sim_\infty G_0 \, \left( 1 - b\, (\tau p)^{-\alpha \gamma}\right)$ implies
that $\kappa(p) = p/c_\infty + \beta(p) \sim_\infty p/c_\infty + p\, \left[ 1 + b\, (\tau p)^{-\alpha \gamma}\right]$
and thus $\beta(p) \sim_\infty (b/c_\infty\, \tau) (\tau p)^{1 - \alpha \gamma}$. On the other hand 
$Q(p) \sim_0 G_\infty\,\left[ 1 + b \gamma/(1 - b \, (\tau p)^\alpha\right]$ and therefore
$\kappa(p) \sim_0  
p/c_0\, \times \\ \left\{1  - b \,\gamma\, (\tau p)^{\alpha + 1}/[2 (1 - b)\,\tau] \right\}$. 
The lowest frequencies propagate with speeds close to $c_0$ and exhibit the attenuation $\mathcal{A}(\omega) \sim_0 b \,\gamma\, \sin(\alpha \upi/2) (\tau \omega)^{1 + \alpha}/[2 (1-b) \tau\, c_0]$, in agreement with experiments. 

Due to its flexibility achieved by a 
minimal number of parameters and asymmetric shape of the peak of the loss modulus the Havriliak-Negami 
relaxation model is frequently used in modeling the alpha relaxation in the mechanical and dielectric data 
in polymers and glass-forming liquids \cite{Boyd85,AligAl88,HavriliakHavriliak94,AlvarezAl93}.

\begin{figure}
\begin{minipage}[t]{0.49\linewidth}
\includegraphics[width=\textwidth]{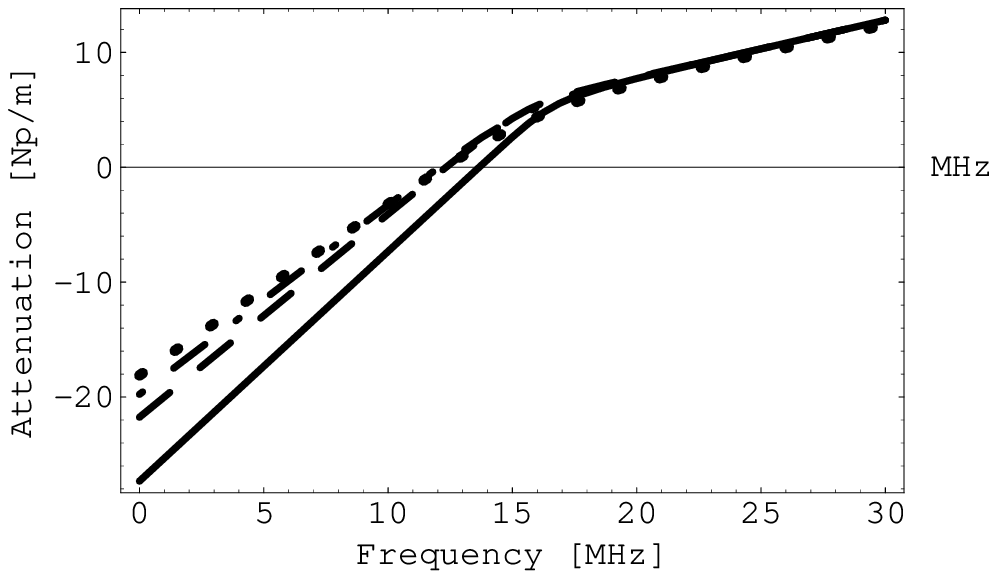}
\begin{center} {\footnotesize (a) Attenuation function.} \end{center}
\end{minipage}\hspace{0.5cm}
\begin{minipage}[t]{0.49\linewidth}
\includegraphics[width=\textwidth]{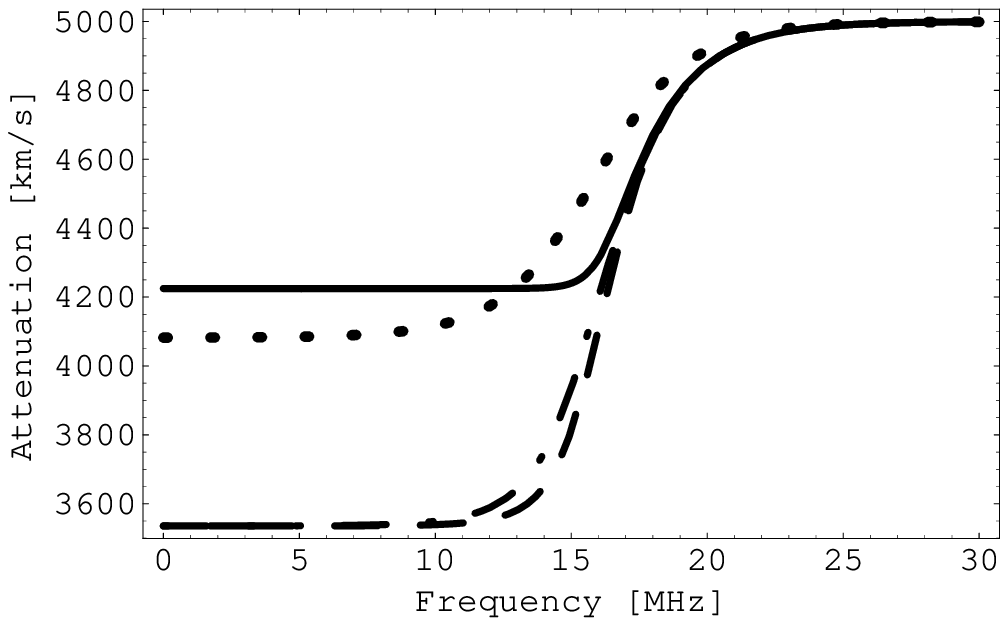}
\begin{center} {\footnotesize (b) Phase speed.} \end{center}
\end{minipage}
\caption{Plots of attenuation and phase speed vs logarithm of angular frequency in MHz
for the Cole-Cole, Havriliak-Negami and Cole-Davidson models, $b = 0.5$, $c_\infty = 5 \, \mathrm{km/s}$, $\tau = 10^{-13} \, \mathrm{s}$.} \label{fig:3}
{\footnotesize Solid line: Cole-Davidson, $\beta=1/2$; dashed line: Havriliak-Negami, $\alpha = 1/1.3$, $\beta=1.3/2$; dot-dashed: 
Havriliak-Negami, $\alpha = 1.3/2$, $\beta=1/1.3$; dotted: Cole-Cole, $\alpha=1/2$.}  \label{fig:2}
\end{figure}

\subsection{The Cole-Davidson relaxation.}

The Cole-Davidson relaxation modulus is given by \eqref{eq:HN} with $\alpha = 1$. The spectral density $h_\gamma^{\mathrm{CD}}$ 
of the Cole-Davidson model cannot however be obtained by substituting $\alpha = 1$ in eq.~\eqref{eq:hHN}. 
For $\alpha = 1$ the branch cut in the complex $p$-plane is reduced
to the segment $]-\infty,-1/\tau]$ and $$Y = \begin{cases} (\tau \,r - 1)^\gamma\, \exp(\ii \upi \gamma),  & \tau\, r > 1\\
(\tau\, r - 1)^\gamma, & \tau \,r \leq 1
\end{cases} $$
Hence 
$$\Re Z = 1 - \begin{cases} b\,\cos(\upi \gamma)\,(\tau \,r  - 1)^{-\gamma}, & \tau\, r > 1\\
b\, (1 - \tau r)^{-\gamma}, & \tau\, r \leq 1 \end{cases}
 , \quad\Im Z = \begin{cases} b\,\sin(\upi \gamma) \, (\tau\, r - 1)^{-\gamma}, & \tau\, r > 1\\
0, & \tau\, r \leq 1\end{cases} 
$$
Consequently
\begin{equation} \label{eq:hCD}
h^{\mathrm{CD}}_\gamma(r) =  \frac{1}{\upi \sqrt{2} c_\infty} 
\frac{\sqrt{k_1(r) - 1 + b \cos(\upi \gamma)}}{k_1(r)} \,\theta(\tau\, r - 1)
\end{equation}
with
$k_1(r) := \sqrt{1 + b^2/\vert 1 - \tau\, r\vert^{2 \gamma} - 2 b \cos(\upi \gamma)/\vert 1 - \tau r\vert^\gamma}$.

The low-frequency asymptotics of the Cole-Davidson attenuation function can be calculated by substituting eq.~\eqref{eq:HN}
in \eqref{eq:kp} and using the definition \eqref{eq:frAtt} of $\mathcal{A}$  yields a behavior characteristic 
of a spectrum which does not extend to 0: 
$\mathcal{A}(\omega) \sim_0 \gamma \, b \, \tau\, \omega^2/[2 c_0 \, (1 - b)] $ 
(cf Sec.~\ref{sec:fb}). 

Ultrasonic data for polymers and bio-tissues cover the range 0--250 MHz, which is much below 
the inverse of the characteristic 
relaxation time $1/\tau = 10^7 \, \mathrm{MHz}$ for polymers. Therefore they are expected to be matched by the 
low-frequency asymptotics of the attenuation function and phase speed. Attenuation and phase speed for the Cole-Cole, Havriliak-Negami and Cole-Davidson models in the low-frequency range 
are shown in Fig.~\ref{fig:2}. The plots have been obtained by numerical integration of eqs~(\ref{eq:frAtt}--\ref{eq:frDisp}) 
using eqs~\eqref{eq:hHN}, \eqref{eq:hCD} and \eqref{eq:c0}.
The asymptotic high-frequency exponent equals 1/2 for all these cases.

\section{Plane waves and minimum phase signals.}
\label{sec:MPH}

The K-K dispersion relations are closely connected to the minimum phase properties 
of the propagators. 

\begin{definition}
A causal tempered distribution $T$ is said to be {\em minimum phase} if its Fourier transform $\hat{T}$
and its inverse $1/\hat{T}(\omega)$ are analytic in the upper half-plane $\mathbb{C}^+$. 
\end{definition}
The first condition implies that the poles and branch cuts of $\hat{T}$ lie in the lower half 
$\omega$ plane.
The second condition implies that $\hat{T}$ does not vanish in the upper half plane $\mathbb{C}^+$.

A shifted function or distribution $\mathcal{T}_\tau f$ is defined by the formula
$(\mathcal{T}_\tau f)(t) = f(t - \tau)$. Its Fourier transform 
$\widehat{\mathcal{T}_\tau f}(\omega) = \hat{f}(\omega) \, \e^{\ii \omega \tau}$ has the same zeros  
and singularities as $\hat{f}(\omega)$. 

The following proposition is obvious.
\begin{proposition} \label{prop:minphaseconv}\mbox{ }\\
\begin{enumerate}[(i)]
\item If the tempered distributions $T$ and $S$ are minimum phase then their convolution is a minimum phase distribution. 
\item If the convolution of the tempered distribution $S$ with an arbitrary minimum phase tempered  
distribution is a minimum phase tempered distribution, then $S$ is a minimum phase tempered
distribution. 
\item If the distribution $T$ is minimum phase, then the shifted distribution $\mathcal{T}_\tau T$ is minimum 
phase distribution. 
\end{enumerate}
\end{proposition}
The convolution $T\ast S$ can be viewed as an application of a linear operator $T$ 
to a signal $S$. A minimum phase operator is defined as a linear operator that preserves the class of 
minimum phase tempered distributions. 
It follows from Proposition~\ref{prop:minphaseconv} that a minimum phase tempered distribution $M$ 
can be viewed either as a minimum phase signal or a minimum phase linear convolution operator (filter). 

\begin{theorem} \label{thm:oup}
A tempered distribution $T$ is minimum phase if and only if $\ln(\hat{T})$ is analytic in the upper half 
plane $\mathbb{C}^+$.
\end{theorem}
\begin{proof}
$\ln(\hat{T})$ is analytic in $\mathbb{C}^+$ if and only if $\hat{T}$ is analytic in $\mathbb{C}^+$
 and does not vanish there. On the other hand, if $\hat{T}$ is analytic in $\mathbb{C}^+$, then
$1/\hat{T}$ is analytic in $\mathbb{C}^+$ if and only if $\hat{T}$ does not vanish there. 
The thesis follows by comparison of the last two statements.
\end{proof} 

Let the distribution or function $f$ be causal. Its Fourier transform $\hat{f}$ is an analytic function in 
$\mathbb{C}^+$. Assume that $\Re \hat{f}(\omega) \rightarrow \infty$ for $\vert \omega \vert \rightarrow 
\infty$ in  $\omega \in \mathbb{C}^+$ is an increasing
function of $\vert \omega \vert$ in $\mathbb{C}^+$, with $\Re \hat{f}(\omega) = 
\OO\left[\vert \omega \vert^\gamma\right]$
for some $\gamma > 0$.
The functions $\e^{\pm \hat{f}(\omega)}$ are analytic in $\mathbb{C}^+$ and 
the function $\e^{-\hat{f}(\omega)}$ tends to zero as $\vert \omega \vert \rightarrow \infty$ in 
$\mathbb{C}_+$. The inverse Fourier transform
$$g(t) = \frac{1}{2 \upi} \int_{-\infty}^\infty \e^{-\ii \omega t} \, \e^{-\hat{f}(\omega)} \, \dd \omega$$
is absolutely convergent and hence a continuous function of $t$. 
By Jordan's lemma the function $g(t)$ is causal. It is also a minimum phase distribution.

\begin{theorem}
Let $T$ be a minimum phase tempered distribution. Then $\ln(\hat{T})$ satisfies the K-K dispersion
relations with two subtractions.
\end{theorem}
\begin{proof} (sketchy)
By Theorem~\ref{thm:oup} $\ln(\hat{T})$ is an analytic function in $\mathbb{C}^+$. Since
$T$ is causal, by the Paley-Wiener Theorem (theorem~\ref{thm:Paley-Wiener}) $\ln(\hat{T}(\omega)) =
\oo_\infty[\omega]$. The function $H(\omega) := \ln(\hat{T}(\omega))/\omega^2$ is analytic in $\mathbb{C}^+$ 
and it uniformly tends to 0 as $\omega$ tends to $\infty$. Hence $H$ is the Fourier transform 
of a causal function $h(t)$. Since $H(\omega) = \oo\left[\omega^{-1}\right]$, the inverse Fourier transform of $H$
is absolutely convergent, the function $f \in \mathcal{C}^0$. Hence $\D^2\, f$ is a distribution of second
order and $\ln(\hat{T})$ satisfies the K-K dispersion relations with two subtractions.
\end{proof}

The Fourier transform of the Green's function $u^{(1)}(t,x)$  
is the product of three factors
\begin{equation} \label{eq:yuyu}
\hat{u}^{(1)}(\omega,x) = \frac{1}{2 k(-\ii \omega)} \, \e^{\ii \omega x/c_\infty} \, \e^{-\beta(-\ii \omega) x}
\end{equation}
The inverse Fourier transform of \eqref{eq:yuyu} is a convolution
\begin{equation} \label{eq:shiftedminphase}
u(t,x) := \delta(t-x/c_\infty) \ast_t \, M(t,x) \ast_t s(t) \equiv M(t-x/c_\infty,x) \ast_t S(t)
\end{equation}
where $S(t)$ is the inverse Fourier transform of the first factor and $M(\cdot,x)$ is 
a minimum phase signal for each $x > 0$. 
The two first factors can be viewed as transfer functions or filters applied to the source signal.

Since $\kappa(p)$ is a CBF function, the singularities of the function 
$k(-\ii \omega) \equiv \ii \kappa(-\ii \omega)$ lie on the negative
imaginary axis. By Corollary~\ref{cor:zerosCBF} its zeros also lie on the negative imaginary axis.
Hence the first factor is a minimum phase filter. 

The third factor has two crucial properties 
\begin{enumerate}
\item it is the Fourier transform of a causal tempered distribution $M$;
\item the real and imaginary parts of the phase $-\beta(-\ii \omega)\, x$ of the Fourier transform 
$\hat{M}$ of $M$ satisfy the K-K dispersion relations.
\end{enumerate} 
These two properties identify $M(t)$ as a minimum phase signal. 
Note that by definition a minimum phase signal has an onset at zero time. In seismology and acoustics 
the onset time of the minimum phase signal can be arbitrary. The signal in \eqref{eq:shiftedminphase} 
is shifted and it starts at $t = x/c_\infty$, the travel time to the point $x$.

The second factor in \eqref{eq:yuyu} is an all-pass filter because it has a constant unit amplitude. It is 
a shift operator acting on the third factor. A shifted minimum phase function is again a minimum phase function.

Hence $u^{(1)}(\cdot,x)$ is a minimum phase signal. If the source function $s(t)$ is a minimum phase signal, then
the convolution $s \ast_t u^{(1)}(\cdot,x)$ is a minimum phase signal for every $x$.

\section{Concluding remarks.} 

We have demonstrated the utility of the spectral decomposition of the wave number function for investigating
general properties of the attenuation and dispersion in  a viscoelastic medium with a positive relaxation spectrum.
It has also been shown that spectral decomposition of attenuation provides a very efficient method for numerical 
computation of attenuation in some popular viscoelastic models of polymers and soft matter. 

Theorem~\ref{thm:rebeta} implies that the complex analytic continuation $\mathcal{A}(\omega)$ of the attenuation function, 
defined on the complex plane except on the negative imaginary axis, is a Herglotz function. Weaver and Pao \cite{WeaverPao81} 
derived the K-K dispersion relations from this property. The derivation of the K-K dispersion relations adopted in this paper 
is based on a constitutive assumption. Weaver and Pao's assumption is not entirely based on a physical argument because it 
concerns the wave number function at complex-valued  frequencies. 

\section*{References.}


\appendix

\section{The Paley-Wiener Theorem}
\label{app:Paley-Wiener}

How can one verify whether a candidate function $M(\omega)$ for the complex 
modulus is the Fourier transform of a function vanishing for $t < 0$?

The best available answer to this question is Theorem XII in \citet{PaleyWiener}:
\begin{theorem} \label{thm:Paley-Wiener}
Let the function $M(\omega)$ be square integrable. \\ If 
$$\int_{-\infty}^\infty \frac{\vert \ln \vert M(\omega)\vert \vert}{1 + \omega^2}
 \, \dd \omega < \infty,$$
then there exists a function $G: \mathbb{R} \rightarrow \mathbb{C}$ with support  
in the closed positive real semi-axis $\overline{\mathbb{R}_+}$ 
such that $M$ is the Fourier transform of $G$.
The converse of this implication is also true.
\end{theorem}
Note that the hypotheses of the theorem involve conditions on the absolute value of $M(\omega)$ 
only. 

\section{Regular variation. Valiron's and Karamata's theorems.}
\label{app:Valiron}

\begin{definition}
A real function $f$ is said to be 
regularly varying at $a$, where $a = 0$ or $\infty$, if
$\lim_{x\rightarrow a} f(\lambda x)/f(x)$ is finite. 

$f$ is said to be rapidly varying 
with index $\infty$ if 
$$\lim_{x\rightarrow a} f(\lambda x)/f(x) = \begin{cases} \infty, & \lambda > 1 \\
0, & \lambda < 1 \end{cases}$$

A real function $f$ on $\mathbb{R}_+$ is said to be rapidly varying at $a$, where $a = 0$ or $\infty$,
with index $-\infty$ if 
$$\lim_{x\rightarrow a} f(\lambda x)/f(x) = \begin{cases} 0, & \lambda > 1 \\
\infty, & \lambda < 1 \end{cases}$$
\end{definition}

\begin{theorem}
If $f$ is regularly varying at $a$, where $a = 0$ or $\infty$, then
$\lim_{x\rightarrow a} f(\lambda x)/f(x) = \lambda^\alpha$,
where $\alpha$ is a real number.   
\end{theorem}
The number $\alpha$ is called the index of the function $f$ at $a$.

\begin{definition}
A real function $f$ on $\mathbb{R}_+$ is said to be slowly varying at $a$ if it
is regularly varying at $a$ with index 0.
\end{definition}

A regularly varying function $f$ with index $\alpha$ at $a$ can be expressed as the
product $x^\alpha \, l(x)$, where $l$ is slowly varying at $a$. 

Examples: $\ln(1 + x)$ is slowly
varying at infinity; $(1 + x^\alpha)$ is slowly varying at 0 if $\alpha \leq 0$ and regularly varying at infinity
with index $\alpha$ if $\alpha > 0$; $\exp(-x)$ is rapidly varying at infinity with index $-\infty$.

\begin{theorem} (Valiron 1911, see \citep{Shea69}) \label{thm:Valiron}
If $F$ is an increasing function satisfying the condition $\lim_{t\rightarrow 0-} F(t) = 0$ and the function $\phi$ is
given by the Stieltjes integral 
$$\phi(r) = \int_{[0,\infty[} \frac{1}{t+r} \, \dd F(t)$$
then the following two statements are equivalent:
\begin{enumerate}[(i)]
\item $F(t) \sim_\infty t^\lambda \,l(t)$;
\item $\phi(r) \sim_\infty \frac{\upi \lambda}{\sin(\upi \lambda)} r^{\lambda-1} \, l(r)$
\end{enumerate}
where $0 \leq \lambda < 1$ and the function $l$ is slowly varying at infinity.
\end{theorem}
Valiron's Theorem also holds when $\infty$ is replaced by 0.

\begin{theorem} \label{thm:KaramataAbelianInfinity} 
(\citet{Feller}, Chapter XIII; \citet{Seneta}, Theorem~2.3)\\
If $F$ is non-decreasing right-continuous, $l$ is slowly varying at 0, $\gamma \geq 0$,
and the Laplace-Stieltjes transform $\phi(p)$ of $F$ converges for $p > a$, 
where $a$ is a 
positive real number, 
then 
$$F(t) \sim_\infty t^\gamma \, l(1/t)/\Gamma(1+\gamma)$$
implies
$$\phi(p) \sim_0 p^{-\gamma} \, l(p)$$ 
\end{theorem}
Note that the function $l(1/t)$ is slowly varying at infinity.

\section{Stieltjes functions. Zeros of a CBF.}
\label{app:zerosCBF}

\begin{definition}
A non-negative real function $f$ on $\mathbb{R}_+$ with the integral representation 
\begin{equation}  \label{eq:StieltjesIntegralRepresentation}
f(x) = a + \frac{b}{x} + \int_{]0,\infty[} \frac{\mu(\dd r)}{x + r}
\end{equation}
where $a, b \geq 0$ and $\mu \in \mathfrak{M}$, is called a {\em Stieltjes function}.
\end{definition}
Eq.~\eqref{eq:StieltjesIntegralRepresentation} can also be expressed in the form
$f(x) = a + \int_{[0,\infty[} \mu(\dd r)/(x + r)$ with $\mu(\{0\}) = b$.

A Stieltjes function $f(x)$ has a complex analytic continuation $f(z)$ in the complex plane cut along the negative real axis,
$\vert \arg(z) \vert < \upi$. The function $f(z)$ has the integral representation 
\eqref{eq:StieltjesIntegralRepresentation} with $x$ replaced by $z$. It follows immediately that $f(z)$ 
maps $\mathbb{C}^+$ to $\mathbb{C}^-$ and vice versa.

\begin{lemma}
A Stieltjes function has a finite value at every point of the complex plane outside 
the closure of the negative real axis.
\end{lemma}
\begin{proof}
Let $ z = x + \ii y$.

The inequality $\vert z + t\vert = \sqrt{(r + x)^2 + y^2} \geq r + x$  
implies for $x > 0$ the inequality 
$$\left \vert \int_{[0,\infty[} \frac{\mu(\dd r)}{z + r} \right \vert \leq \int_{[0,\infty[} \frac{\mu(\dd r)}{x + r} < \infty$$
(cf Remark~\ref{rem:1} ).

If $y > 0$ then we can find such a number $\vartheta$ that $0 < \vartheta \leq 1$ and
$q := \vartheta x + (1 - \vartheta) y > 0$. Hence 
$$\sqrt{(r + x)^2 + y^2}  \geq \vartheta (r + x) + (1 - \vartheta)\, y \geq \vartheta\, (r + q/\vartheta)$$
Hence
$$\int_{[0,\infty[} \frac{\mu(\dd r)}{\vert z + r\vert} \leq \frac{1}{\vartheta} \int_{[0,\infty[} \frac{\mu(\dd r)}{r + q/\vartheta} < \infty$$
\end{proof}

\begin{theorem} \label{thm:anCM} \cite{Akhiezer}\\
An analytic function $f$ regular in $\mathbb{C}^+$ and satisfying the inequality 
(i) $\Im f(z) \leq 0$ in $\mathbb{C}^+$ and (ii) $f(x) \geq 0$ for 
$x \in \mathbb{R}_+$ 
can be expressed in the form 
$$f(z) = f_0 + \int_{[0,\infty[} \frac{\mu(\dd y)}{y + z}$$
where $f_0 = \lim_{z\rightarrow \infty} f(z)$ and $\mu \in \mathfrak{M}$.
The limit $\lim_{z\rightarrow\infty} f(z)$ exists and equals $f_0$. The measure
$\mu$ is uniquely defined as the weak limit 
\begin{equation}\label{eq:wl}
\int_{[0,\infty[} \phi(y) \, \mu(\dd y) = -\frac{1}{\upi} \, 
\lim_{\varepsilon\rightarrow 0+} \int_{[0,\infty[}\ \phi(y) \, 
\Im f(-y + \ii \varepsilon)\,\dd y
\end{equation}
 for every continuous function $\phi$ with compact support.
\end{theorem}

\begin{theorem} \label{thm:invStCBF}
The inverse $1/f$ of a non-zero CBF $f$ is a Stieltjes function.
\end{theorem}
\begin{proof}
The integral representation \eqref{eq:CBFintegral} implies that
$$f(1/z) = a_1 + b \frac{1}{z} + \int_{]0,\infty[} \frac{\rho(\dd y)}{1 + z\, y} =
a_1 + b \frac{1}{z} + \int_{]0,\infty[} 
\frac{u \,\rho\left(\dd\, u^{-1}\right)}{u + y}$$
where $a_1 := a + \rho(\{0\}) \geq 0$, $u = 1/y$. Hence the function 
$f(1/z)$ satisfies 
(i) and (ii) of Theorem~\ref{thm:anCM}. In particular, (ii) implies 
that $\Im f(1/z)^{-1} \geq 0$ for 
$z \in \mathbb{C}^+$.
Substitute $z = 1/\zeta$. Clearly, $\sgn \Im \zeta = - \sgn \Im z$ and, 
by reflection,
$\Im f(\zeta)^{-1} \leq 0$ for $\zeta \in \mathbb{C}^+$. Since $f(x)^{-1} > 0$
as well, $f(x)^{-1}$ is a Stieltjes function. 
\end{proof}

\begin{corollary} \label{cor:zerosCBF}
A CBF does not vanish outside the closed negative real half-axis.
\end{corollary}
\begin{proof}
Let $f$ be a non-zero CBF. If $f(a) = 0$ then the Stieltjes function $1/f$ has a pole 
at $a$. Hence $a \in \mathbb{R}_-$. 
\end{proof}

\section{Proof of Theorem~\ref{thm:locint}.}
\label{app:locint}

\begin{proof}
By the Fubini theorem
$$\int_0^1 \dd t \int_{[0,\infty[} \e^{-r t} \mu(\dd r) = 
\int_{[0,\infty[}  \mu(\dd r) \int_0^1 \e^{-r t} \, \dd t = 
\int_{[0,\infty[}   \frac{1 - \e^{-r}}{r} \, \mu(\dd r)$$

If $f$ if integrable over [0,1] then the last integral is convergent.
Since the integrand is non-negative and the Radon measure $\mu$ is positive,
$$\int_{[1,\infty[}   \frac{1 - \e^{-r}}{r} \, \mu(\dd r)  < \infty $$
But 
$$
\int_{[1,\infty[}   \frac{1 - \e^{-r}}{r} \, \mu(\dd r) \geq
\left(1 - \e^{-1}\right) \int_{[1,\infty[} \frac{\mu(\dd r)}{r} \\
\geq \left(1 - \e^{-1}\right)  \int_{[1,\infty[} \frac{\mu(\dd r)}{1 + r}
$$
hence the last integral is convergent. 
The integral 
$$ \int_{[0,1[} \frac{\mu(\dd r)}{1 + r}$$ 
is convergent, hence the inequality \eqref{eq:doss} is satisfied.

Assume now that \eqref{eq:doss} is satisfied. 
Note that $$\int_{[0,1]} \mu(\dd r) \leq 2 \int_{[0,1]} \frac{\mu(\dd r)}{1 + r} < \infty$$
It follows that the integral 
$$\int_{[0,1]}  \frac{1 - \e^{-r}}{r} \, \mu(\dd r)$$ 
is convergent because the integrand is bounded. It remains to consider the
integral on $[1,\infty[\;$:
$$
\int_{[1,\infty[}  \frac{1 - \e^{-r}}{r} \, \mu(\dd r) \leq 
\int_{[1,\infty[}  \frac{1}{r} \, \mu(\dd r) \leq 
2 \int_{[1,\infty[} \frac{\mu(\dd r)}{1 + r}
$$
hence $f$ is integrable over [0,1].
\end{proof}

\section{Proof of the CM property of the relaxation moduli in Sec.~\ref{sec:examples}.}
\label{app:HNCM}

We prove here that $G^{\mathrm{HN}}_{\alpha,\gamma}$ is CM for $0 \leq \alpha, \gamma \leq 1$. According to Theorem~2.6 
in Chapter~16 of \cite{GripenbergLondenStaffans} it suffices to prove that 
$F(p) := \Im\left[p\, \tilde{G}^{\mathrm{HN}}_{\alpha,\gamma}(p)\right] \geq 0$ for $p \in \mathbb{C}^+$ and
$\tilde{G}^{\mathrm{HN}}_{\alpha,\gamma}(p) \geq 0$ for $p \in \mathbb{R}_+$. The second inequality follows from the 
assumption that $b \leq 1$. The first one follows from the fact that $\alpha, \gamma \leq 1$, $p \in \mathbb{C}^+$ imply that 
$\left(1 + (\tau\, \overline{p})^\alpha\right)^\gamma \in \mathbb{C}^-$ and therefore 
$F(p) = -b \Im \left[(1 + (\tau \overline{p})^\alpha)^\gamma\right]/\left\vert 1 + (\tau \, p)^\alpha\right\vert^{2 \gamma} 
\geq 0$. The remaining hypotheses of Theorem~2.6 are easy to verify. 

The CM property holds for $\alpha \leq 1$ and $\gamma \leq 1$, hence it also applies to $G^{\mathrm{CD}}_{\gamma}$
and $G^{\mathrm{CC}}_{\alpha}$. 

\end{document}